%% file: paper.tex
\documentclass[11pt, a4paper]{article} % don't like fleqn -- kr

\input{preamble.tex}

\input{notation.tex}

\title{Sum-of-Squares Lower Bounds for\\the Minimum Circuit Size
    Problem\thanks{ Supported by the Approximability and Proof
    Complexity project funded by the Knut and Alice Wallenberg
    Foundation and the Swiss National Science Foundation project
    200021-184656 “Randomness in Problem Instances and Randomized
    Algorithms”.}
    }

\author{Per Austrin}
\affil{KTH Royal Institute of Technology}

\author{Kilian Risse}
\affil{EPFL}

\date{\today}

\begin{document}

\maketitle
\thispagestyle{empty}
\addtocounter{page}{-1}

\input{abstract.tex}

\clearpage

% while writing
%% \tableofcontents

\input{introduction.tex}

\input{results.tex}

\input{preliminaries.tex}

\input{circuits-and-restrictions.tex}

\input{de-morgan-lb.tex}

\input{monotone-lb.tex}

\input{upper-bound.tex}

\input{conclusion.tex}

\bibliographystyle{alpha}
\bibliography{references}

\clearpage

\appendix

\input{cnf-encoding}

\end{document}

%% file: preamble.tex
\usepackage[sc]{mathpazo}
\usepackage{bm}

\usepackage{amsmath}
\usepackage{enumitem}
\usepackage{amssymb}
\usepackage{amsfonts}
\usepackage{textcomp}
\usepackage{algorithm}
\usepackage{authblk}
\usepackage[noend]{algpseudocode}
\usepackage{amsthm}
\usepackage{xspace}
\usepackage{fullpage}
\usepackage[english]{babel}
\usepackage[utf8]{inputenc}
\usepackage[classfont = bold]{complexity}
\usepackage{mathtools}
\usepackage[pagebackref=true, colorlinks=true, citecolor=DarkGreen, linkcolor=Navy]{hyperref}
\usepackage[svgnames]{xcolor}
\usepackage{bussproofs}
\usepackage{subfig}
\usepackage{todonotes}
\usepackage{thmtools}
\usepackage{thm-restate}
\usepackage[capitalise, noabbrev, nameinlink]{cleveref}

\theoremstyle{plain}
\newtheorem{theorem}{Theorem}[section]
\newtheorem{lemma}[theorem]{Lemma}

\newtheorem{proposition}[theorem]{Proposition}
\newtheorem{claim}[theorem]{Claim}

\theoremstyle{remark}

\theoremstyle{definition}
\newtheorem{definition}[theorem]{Definition}

\newcommand{\pa}[1]{}
\newcommand{\kr}[1]{}

%% file: notation.tex
\newcommand\defeq{\stackrel{\mathclap{\normalfont\mbox{\scriptsize{def}}}}{=}}

\newcommand{\eps}{\varepsilon}

\newcommand{\eqcomma}{\enspace ,}
\newcommand{\eqperiod}{\enspace .}

\newcommand{\set}[1]{\{#1\}}
\newcommand{\setdescr}[2]{\{\,#1 \mid #2\,\}}

\newcommand\restrict[2]{{% we make the whole thing an ordinary symbol
  \left.\kern-\nulldelimiterspace % automatically resize the bar with \right
  #1 % the function
  \vphantom{\big|} % pretend it's a little taller at normal size
  \right|_{#2} % this is the delimiter
  }}
\usepackage{bbm}

\newcommand{\calC}{\mathcal{C}}

\newcommand{\calF}{\mathcal{F}}

\newcommand{\calM}{\mathcal{M}}

\newcommand{\calP}{\mathcal{P}}
\newcommand{\calQ}{\mathcal{Q}}

\newcommand{\calT}{\mathcal{T}}

\renewcommand{\R}{\mathbb{R}}
\newcommand{\N}{\mathbb{N}}

\DeclareMathOperator{\Deg}{\sf{Deg}}
\DeclareMathOperator{\Size}{\sf{Size}}
\newcommand{\RefuteDeg}[2]{\Deg(#2 \vdash_{\text{\tiny #1}} \, \perp)}
\newcommand{\RefuteSize}[2]{\Size(#2 \vdash_{\text{\tiny #1}} \, \perp)}

\newcommand{\mon}{{\mathsf{mon}}}

\DeclareMathOperator{\csize}{\calC}
\DeclareMathOperator{\cmonsize}{\csize_{\mathsf{mon}}}

\newcommand{\gout}{g^\mathrm{out}}
\newcommand{\ginone}{g^\mathrm{in_1}}
\newcommand{\gintwo}{g^\mathrm{in_2}}

\DeclareMathOperator{\neigh}{\textbf{Neigh}}

\newcommand{\subformula}{$ is a variable substitution of $}

\DeclareMathOperator{\sel}{\mathsf{Sel}}

\newcommand{\vars}[1]{\mathsf{#1}}
\DeclareMathOperator{\circuit}{Circuit}
\DeclareMathOperator{\outwire}{\vars{Out}}
\DeclareMathOperator{\inwire}{\vars{In}}
\DeclareMathOperator{\isneg}{\vars{IsNeg}}
\DeclareMathOperator{\isand}{\vars{IsAnd}}
\DeclareMathOperator{\isor}{\vars{IsOr}}
\DeclareMathOperator{\isfromconst}{\vars{IsFromConst}}
\DeclareMathOperator{\isfromvar}{\vars{IsFromInput}}
\DeclareMathOperator{\isfromgate}{\vars{IsFromGate}}
\DeclareMathOperator{\constval}{\vars{ConstantValue}}
\DeclareMathOperator{\isvar}{\vars{IsInput}}
\DeclareMathOperator{\isgate}{\vars{IsGate}}

\crefname{axiom}{Axiom}{Axioms}
\Crefname{axiom}{Axiom}{Axioms}

\crefname{axioms}{Axioms}{Axioms}
\Crefname{axioms}{Axioms}{Axioms}
\DeclareMathOperator{\isvarless}{\vars{IsInput^{\le}}}
\DeclareMathOperator{\isgateless}{\vars{IsGate^{\le}}}

%% file: abstract.tex
\begin{abstract}
  We prove lower bounds for the Minimum Circuit Size Problem (MCSP) in
  the Sum-of-Squares (SoS) proof system.  Our main result is that for
  every Boolean function $f: \{0,1\}^n \rightarrow \{0,1\}$, SoS
  requires degree $\Omega(s^{1-\epsilon})$ to prove that $f$ does not
  have circuits of size $s$ (for any $s > \poly(n)$).  As a corollary
  we obtain that there are no low degree SoS proofs of the statement
  $\NP \not \subseteq \Ppoly$.

  We also show that for any $0 < \alpha < 1$ there are Boolean
  functions with circuit complexity larger than $2^{n^{\alpha}}$ but
  SoS requires size $2^{2^{\Omega(n^{\alpha})}}$ to prove this.
  In addition we prove analogous results on the minimum
  \emph{monotone} circuit size for monotone Boolean slice functions.

  Our approach is quite general.  Namely, we show that if a proof
  system $Q$ has strong enough constraint satisfaction problem lower
  bounds that only depend on good expansion of the constraint-variable
  incidence graph and, furthermore, $Q$ is expressive enough that
  variables can be substituted by local Boolean functions, then the
  MCSP problem is hard for $Q$.
\end{abstract}

%% file: introduction.tex
\section{Introduction}

Even before the dawn of complexity theory, there was an interest in
the minimum circuit size problem (MCSP): given the truth table of a
Boolean function $f: \set{0,1}^n \rightarrow \set{0,1}$ and a
parameter $s$, the MCSP problem asks whether there is a Boolean
circuit of size at most $s$ computing $f$.  Despite many years
of research, we do not know whether this problem is \NP-hard. It
clearly is in \NP: given a circuit of size at most $s$ (described by
$O(s \log s)$ bits) we can easily check in time $O(s \cdot 2^n)$
whether this circuit indeed computes $f$.

Determining the hardness of MCSP itself turns out to be a difficult
problem. Kabanets and Cai \cite{KC00} showed that \NP-hardness of the
MCSP problem implies breakthrough circuit lower bounds. These lower
bounds are not implausible but well out of reach of current
techniques.  In a similar vein Murray and Williams \cite{MW15} showed
that \NP-hardness of MCSP implies that $\EXP \neq \ZPP$ and more
recently Hirahara \cite{Hi18} proved that \NP-hardness of MCSP implies
a worst-case to average-case reduction for problems in \NP~(for an
appropriate MCSP version).

On the other hand if one could show that MCSP is in \Ppoly, this would
imply even stronger (though less realistic) results: Kabanets and Cai
\cite{KC00} also showed that if MCSP is in \Ppoly, then
crypto-secure one way functions can be inverted on a considerable
fraction of their range.

To summarize it seems unlikely that MCSP is in \P, but showing that
it is \NP-hard implies very strong consequences. As these results seem
out of reach for current techniques, it might be a more fruitful
avenue to try to at least rule out that certain (families of) algorithms
solve the MCSP problem efficiently.

This can be achieved very elegantly in proof complexity: show that
some proof system capturing your algorithm requires long proofs to
refute the claim that a complex function has a small circuit. This
will then rule out that the algorithm in question can efficiently
solve the MCSP problem.  This will not only show that this
specific algorithm requires long running time but would also show that
any algorithm captured by this proof system requires long running time
to solve the MCSP problem. Hence by this line of reasoning we manage
to rule out entire classes of algorithms to solve the MCSP problem
efficiently.

This paper focuses on the Sum of Squares proof system (SoS). This
proof system provides certificates of unsatisfiability of systems of
polynomial equations $\calP = \set{p_1 = 0, \ldots, p_m = 0}$ over
$\R$.  In this paper we are only interested in Boolean systems of
equations, meaning that $\calP$ contains the equation $x^2-x=0$ for
every variable $x$ appearing in the system.  A key complexity measure
is the degree of a refutation, which is the maximum degree of a
monomial occurring in the refutation of $\calP$. All Boolean systems
$\calP$ over $n$ variables have an SoS refutation of degree $n$ and we
are interested in the minimum degree that SoS requires to refute
$\calP$.  An SoS refutation of degree $d$ has size $O(n^d)$ and can be
found in $n^{O(d)}$ time using semidefinite programming and this is
often a useful heuristic bound for the complexity of an SoS
refutation.  The actual size complexity of SoS can sometimes be
significantly smaller than $n^{d}$ \cite{potechin20} and it would be
surprising if the shortest refutation can be found efficiently. Hence
it is in general of interest to understand both the degree and the
size needed to refute any given system.

SoS is a very powerful proof system and captures many state of the art
algorithms that are based on spectral methods. A classic algorithm
captured by SoS is Goemans and Williamson's Max-Cut algorithm
\cite{GW95}, but also approximate graph coloring algorithms \cite{KMS98}, and
algorithms solving constraint satisfaction problems \cite{AOW15, RRS17} are
captured by SoS. On the other hand SoS has real difficulty arguing
about integers and in particular parities. Indeed, Grigoriev
\cite{gri01xor} showed that SoS requires degree $\Omega(n)$ to refute
a random \emph{xor} constraint satisfaction problem of the appropriate
(constant) density. After this initial lower bound it took a few years
to develop good lower bounds methods for SoS, but in recent years
there has been a flurry of strong SoS degree lower bounds
\cite{MPW15SumOfSquaresPlantedClique, BHKKMP16clique, kmow17anycsp}.

In order to rule out that algorithms captured by SoS can solve MCSP
efficiently, we need to encode the claim that a given function has a
small circuit as a propositional formula. We work with the encoding
suggested by Razborov \cite{Razborov98}, which encodes this claim that
the function $f: \set{0,1}^n \rightarrow \set{0,1}$ has a circuit of
size $s$ by a propositional formula $\circuit_s(f)$ over
$O(s^2 + s \cdot 2^n) = O(s \cdot 2^n)$ variables as follows. We have
$\Theta(s^2)$ \emph{structure variables} to encode all possible size
$s$ circuits, and for every assignment $\alpha \in \set{0,1}^n$ we
then have an additional $\Theta(s)$ \emph{evaluation variables} that
simulate the evaluation of the circuit on each input, and constraints
forcing the circuit to output the correct value on each input
$\alpha$.

A closely related question to the MCSP problem is the question of how
hard it is to actually prove strong circuit lower bounds. For example,
are there efficient refutations of the statement
$\NP \subseteq \Ppoly$, assuming the statement is false? This
question, as proposed by Razborov \cite{Razborov98}, can also be
investigated by studying above formula: consider
$\circuit_{n^{O(1)}}(\SAT)$, where $\SAT$ is the function that outputs
$1$ if and only if the input is an encoding of a satisfiable CNF. This
is, essentially, a propositional encoding of the claim that $\SAT$ has
a circuit in $\Ppoly$. Hence proving lower bounds for
$\circuit_{n^{O(1)}}(\SAT)$ rules out efficient proofs of
$\NP \not\subseteq \Ppoly$ in the proof system under consideration.

Experience suggests that studying such meta-mathematical questions is
difficult. This problem is no exception to this rule and, even though
the formula has been conjectured to be hard for strong proof systems
such as extended Frege, progress has been slow. The only proof systems
for which we have unconditional, superpolynomial lower bounds on
proofs of the $\circuit_s(f)$ formula are Resolution
\cite{Raz04,Razborov04ResolutionLowerBoundsPM}, small width
DNF-Resolution \cite{Razborov15PRG} and Polynomial Calculus
\cite{Razborov98, Razborov15PRG}. The resolution size and Polynomial
Calculus degree lower bounds follow from a reduction of the pigeonhole
principle to $\circuit_s(f)$. In fact, this reduction was a main
motivation for a long line of work \cite{RazbWigYao02, PitRaz04,
  Raz04, Razborov04ResolutionLowerBoundsPM} eventually establishing
strong resolution lower bounds for the weak pigeonhole principle.  The
other size lower bounds follow from a general connection between
pseudo-random generator lower bounds and MCSP lower bounds as outlined
in \cite{ABRW04Pseudorandom,Razborov15PRG}, building on Krajíček's
iterability trick \cite{Krajicek01}.

As the pigeonhole principle is easy for the SoS proof system
\cite{GHA02proofs}, we cannot hope to borrow the hardness from that
formula. Neither do we have strong enough pseudorandom generator lower
bounds for SoS to employ that connection. In fact, to date, we have no
unconditional (degree) lower bounds for any semi-algebraic proof
system, that is, proof systems that manipulate polynomial inequalities
such as SoS or Cutting Planes. Furthermore it has been stated
\cite{razborov21youtube, razborov22website} as an explicit open
problem to prove SoS degree lower bounds for the formula
$\circuit_s(f)$.

%% file: results.tex
\subsection{Our Results}

Our first result gives a lower bound on the degree needed to refute
$\circuit_s(f)$ in SoS.  This lower bound is very general and in fact
applies to \emph{every} Boolean function
$f: \set{0,1}^n \rightarrow \set{0,1}$.

\begin{restatable}{theorem}{MainDegree}
  \label{thm:main-degree}
  For all $\eps > 0$ there is a $d = d(\eps)$ such that the following
  holds.  For $n \in \N$, all $s \ge n^{d}$ and any Boolean function
  $f:\set{0,1}^{n} \rightarrow \set{0,1}$ on $n$ bits, SoS requires
  degree $\Omega_\eps(s^{1-\eps})$ to refute $\circuit_s(f)$.
\end{restatable}

It is worthwhile to point out that the proof of \cref{thm:main-degree}
is not specific to the SoS proof system. In fact we outline a general
reduction that shows that if one has a CSP lower bound of the form of
\cref{thm:sos-xor} that only requires good expansion of the underlying
constraint-variable incidence graph and the proof system is expressive
enough so that one can replace variables by local Boolean functions,
then one obtains strong lower bounds for the $\circuit_s(f)$ formula.

The lower bound of $\Omega_\eps(s^{1-\eps})$ on the degree is
essentially tight: if $f$ does not have a circuit of size $s$ then
there exists an SoS refutation of this statement in degree $O(s)$.

\begin{proposition}
  Let $s \in N$ and $f: \set{0,1}^n \rightarrow \set{0,1}$ be a
  Boolean function on $n$ bits that requires circuits of size larger
  than $s$ to be computed. Then there is a degree $O(s)$ SoS
  refutation of $\circuit_s(f)$.
\end{proposition}

We also prove a result about the minimum size (number of monomials)
required for SoS to refute $\circuit_s(f)$. This result holds for all
functions that ``almost'' have a circuit of size $s$, in the sense that
they have an errorless heuristic circuit (see the survey \cite{BT06})
of size $s/2$ and extremely small error probability
with respect to the uniform distribution.
Formally, we let $\calF_n(s,t)$ denote the class of Boolean functions
that consists of all functions $f: \set{0,1}^n \rightarrow \set{0,1}$
for which there is a Boolean circuit
$C_f: \set{0,1}^n \rightarrow \set{0, 1, \bot}$ of size at most $s$
such that
\begin{enumerate}
\item if $C_f(\alpha) \neq \bot$, then $C_f(\alpha) = f(\alpha)$, and
\item $C_f(\alpha) = \bot$ on at most $t$ inputs.
\end{enumerate}
In other words the circuit $C_f$ computes $f$ correctly on all except
$t$ inputs. Note that technically the output of the circuit $C_f$ is
two bits with the first one indicating whether the output is $\bot$ or
the value of the second bit.  We believe that above presentation is
more intuitive and hope that the slight abuse of notation causes no
confusion. With the class of functions $\calF_n(s,t)$ at hand we can
state our main SoS size lower bound.

\begin{restatable}{theorem}{MainSize}
  \label{thm:main-size}
  For all $\eps > 0$ there is a $d = d(\eps)$ such that the following
  holds. Let $n \in N$ and $s \in \N$ such that $s \ge n^d$. If
  $t \ge s$ and $f \in \calF_n(s/2, t)$, then it holds that SoS
  requires size $\exp\big(\Omega_\eps(s^{2-\eps}/t)\big)$ to refute
  $\circuit_s(f)$.
\end{restatable}

This yields non-trivial size lower bounds for $t$ as large as
$s^{2-\eps}/\omega(1)$.  Furthermore, note that once $t \gg s \log s$
there are functions that require such large circuits.  For example
setting $s = 2^{n^{0.99}}$ and $t = s^{1.5}$, the theorem shows that
there are functions $f$ that do not have circuits of size $s$, but SoS
requires size $2^{2^{\Omega(n^{0.99})}}$ to prove this.

It is natural to wonder whether SoS fares better in the monotone
setting. In other words, whether SoS can refute the claim that a
complex monotone function has a small monotone circuit. The following
two theorems show that this is not the case for the set
$\calM_n(\ell)$ of monotone $\ell$-slice functions.  Recall that
$\calM_n(\ell)$ consist of all Boolean functions $f$ on $n$ bits such
that $f(\alpha) = 0$ for all $\alpha$ with Hamming weight less than
$\ell$, and $f(\alpha) = 1$ for all $\alpha$ with Hamming weight
greater than $\ell$ (note that any such $f$ is monotone).

We define a variant $\circuit_s^{\mon}(f)$ of the $\circuit_s(f)$
formula, which instead encodes the claim that $f$ has a monotone
circuit of size $s$, and prove the following theorem.

\begin{restatable}{theorem}{MonotoneDegree}
  \label{thm:main-monotone-degree}
  For all $\eps > 0$ there is a $d = d(\eps)$ such that the following
  holds.  For all $n, \ell \in \N$, all $s \ge n^{d}$ and any monotone
  slice function $f \in \calM_n(\ell)$ SoS requires degree
  $\Omega_\eps(s^{1-\eps})$ to refute
  $\circuit_s^{\mon}(f)$.
\end{restatable}

As in the non-monotone case, we can also obtain size lower bounds for
the monotone-MCSP. Akin to the general size lower bound we consider
monotone Boolean slice functions that have good monotone errorless
heuristic circuits. Let $\calM_n(\ell, s, t) \subseteq \calM_n(\ell)$
be the class of monotone Boolean $\ell$-slice functions
$f: \set{0,1}^n \rightarrow \set{0,1}$ for which there is a
(not necessarily monotone) Boolean circuit
$C_f^\mon: \set{0,1}^n \rightarrow \set{0,1, \bot}$ of size $s$ such
that
\begin{enumerate}
\item for all $\ell$-slice inputs $\alpha \in \binom{[n]}{\ell}$ it
  holds that if $C^\mon_f(\alpha) \neq \bot$, then
  $C^\mon_f(\alpha) = f(\alpha)$, and
\item $C^\mon_f(\alpha) = \bot$ on at most $t$ inputs
  $\alpha \in \binom{[n]}{\ell}$.
\end{enumerate}

\begin{restatable}{theorem}{MonotoneSize}
  \label{thm:main-monotone-size}
  For all $\eps > 0$ there is a $d = d(\eps)$ such that the following
  holds.  For $n, \ell \in \N$, all $s \ge n^{d}$ and $t \ge s$ and
  monotone function $f \in \calM_n(\ell, s/10, t)$ SoS requires size
  $\exp\big(\Omega_\eps(s^{2-\eps}/t)\big)$ to refute
  $\circuit_s^{\mon}(f)$.
\end{restatable}

\subsection{Overview of Proof Techniques}

\paragraph{Degree Lower Bound:}
The main idea that drives our result is a reduction from an expanding
\emph{xor} constraint satisfaction problem to the $\circuit_s(f)$
formula. The reduction is achieved through a careful restriction of
the $\circuit_s(f)$ formula, such that each input
$\alpha \in \set{0, 1}^n$ to the circuit specifies an \emph{xor}
constraint over some new set of variables $Y$.  These $Y$ variables
are a subset of roughly $\Theta(s^{1-\epsilon})$ out of the
$\Theta(s^2)$ many structure variables of the $\circuit_s(f)$ formula.
All other structure variables apart from the $Y$ variables are fixed
to constant values in this step.
This will then result
in an XOR-CSP instance with $2^n$ constraints over the variables
$Y$. All that SoS has to prove is that there is no satisfying
assignment to this XOR-CSP instance. By ensuring that the
constraint-variable incidence graph is sufficiently expanding, SoS
requires large degree to refute the restricted formula (see
\cref{thm:sos-xor}). At the same time, we need the constraint graph to
be very explicit so that it can be encoded into a small circuit.  For
this we utilize a construction of unbalanced expanders by Guruswami et
al.~\cite{GUV09Unbalanced} (see \cref{thm:expander}).  This reduction
then immediately yields \cref{thm:main-degree}.

\pa{Maybe in the above paragraph we also want to point out that all
  the $s \cdot 2^n$ evaluation variables are still alive, effectively
  giving SoS a bunch of auxiliary variables to use when refuting the
  XOR instance, and that some care needs to be taken to deal with
  this?}

\kr{I think this is too technical at this point? I would rather not
  comment on this.}

This lower bound may also be viewed as implementing the general
program sketched by Razborov \cite{Razborov15PRG} relating pseudorandom
generators in proof complexity to the MCSP problem.  However, we prefer
to describe it as a direct reduction to the MCSP problem.

\paragraph{Size Lower Bound:}
In order to obtain size lower bounds, we would like to apply the
degree-size tradeoff due to Atserias and Hakoniemi
\cite{AH18SosTradeoff} to 
\cref{thm:main-degree}. Unfortunately the formula is
over too many variables to be able to conclude a meaningful size lower
bound: it is defined over roughly $\Omega(2^n \cdot s)$
variables.

Instead of applying \cref{thm:main-degree}, we restrict our attention
to functions with all except the at most $t$ $\bot$-outputs computed
by the corresponding errorrless heuristic circuit. If we choose $t$
small enough, then we are able to heavily restrict $\circuit_s(f)$ and
significantly reduce the number of variables to the point where the
Atserias-Hakoniemi degree-size tradeoff is applicable.

\paragraph{Monotone Circuits:}

We prove these theorems by adapting the proofs for the non-monotone setting.
The idea is to work over the $\ell$th slice and disregard all other
inputs. The key feature that makes this work is the fact that the
monotone circuit complexity of a slice function is essentially the same as
the (ordinary) circuit complexity (see \cref{lem:slice}).  This lets us
convert all subcircuits used in the reduction to small monotone circuits (if
we only work on the slice).
% We further need to take care of negations
% in the part of the circuit that computes the \emph{xor}. We do this by
% pushing the negations down until they either hit a $Z$ variable or one
% of the prepared monotone slice circuits.  We create a set
% $\overline Z$ gates, which we can think of as the negation of the
% gates in $Z$ and also make negated versions (on the slice) of the
% monotone circuits that are required for the reduction. This allows us
% to get rid of the last negations in the circuit by appropriately
% connecting the different subcircuits.

The size lower bound goes along the same lines as the proof of
\cref{thm:main-size}. 

\subsection{Organization}

In \cref{sec:preliminaries}, we provide the necessary background
material. In \cref{sec:circ-rest} we set up the general framework for our
lower bounds with some preliminary definitions and lemmas.
Then in \cref{sec:demorgan-LB} we prove the main degree \cref{thm:main-degree} and
size \cref{thm:main-size} lower bounds. We
prove the monotone lower bounds \cref{thm:main-monotone-degree} and
\cref{thm:main-monotone-size} in \cref{sec:monotone-degree-LB}.  In
\cref{sec:upper-bound} we explain how SoS of degree $O(s)$ can refute
$\circuit_s(f)$ (provided $f$ does not have a circuit of size $s$).
Finally in \cref{sec:conclusion} we give some concluding remarks.

%% file: preliminaries.tex
\section{Preliminaries}
\label{sec:preliminaries}

All logarithms are in base $2$.  For integers $n \ge 1$ we write
$[n] = \set{1, 2, \ldots, n}$ and for a set $U$ we denote the power
set of $U$ by $2^{U}$. Further, for a set $V \subseteq U$ we let
$\overline V$ be the complement of $V$ with respect to $U$, that is,
$\overline V = U \setminus V$. We write
$\binom{[n]}{\ell} \subseteq \set{0,1}^n$ for the set of binary
strings with Hamming weight $\ell$. For a string
$\alpha \in \set{0,1}^n$ we let
$|\alpha| = \sum_{i \in [n]} \alpha_i$.

We sometimes want to supress dependencies on constants and write
$f(n, \eps) \in O_\eps\big(g(n, \eps)\big)$, respectively
$f(n, \eps) \in \Omega_\eps\big(g(n, \eps)\big)$, to mean that there
exists a function $c(\eps) > 0$ such that there is an $n_0$ and for all
$n \ge n_0$ it holds that $f(n, \eps) \le c(\eps) \cdot g(n, \eps)$,
respectively $f(n, \eps) \ge c(\eps)\cdot g(n, \eps)$.

\begin{definition}\label{def:explicit}
  A sequence of bipartite graphs $\set{G_n=(U_n, V_n, E_n)}_{n\in \N}$
  with $\deg(u) =d$ for all $u \in U_n$ is \emph{explicit} if there is
  an algorithm that given $(n, u, j)$, where $n \in \N, u \in U_n$ and
  $j\in [d]$, computes the $j$th neighbor of vertex $u$ in the graph
  $G_n$ in time $\poly(\log n + \log |U| + \log d)$.
\end{definition}

From now on it is understood that whenever we talk about an explicit
graph we actually mean to say that there is a sequence of explicit
graphs with above properties.

For a graph $G = (V, E)$ and $W \subseteq V$ we denote by
\[
  N(W) = \setdescr{u \in V \setminus W}{(w,u) \in E \text{ for some } w
    \in W}
\]
the set of neighbors of $W$.

\begin{definition}
  A bipartite graph $G = (U, V, E)$ is an $(r, d, c)$-expander if every
  vertex $u \in U$ has degree $\deg(u) = d$ and every set
  $W \subseteq U$ of size $|W| \le r$ satisfies $|N(W)| \ge c \cdot |W|$.
\end{definition}

A key ingredient in our proofs is the following result on the
existence of strong explicit expanders.

\begin{theorem}[\cite{GUV09Unbalanced}]\label{thm:expander}
  For all constants $\gamma > 0$,
  every $M \in \N$, $r \le M$,
  and $\eps > 0$,
  there is an
  $N \le d^2 \cdot r^{1 + \gamma}$ and an
  explicit
  $(r, d, (1- \eps)d)$-expander
  $G=(U, V, E)$,
  with $|U| = M$,
  $|V| = N$, and
  $d = O\big( ((\log M)(\log r) / \eps )^{1+1/\gamma}\big)$.
\end{theorem}

\kr{Igor pointed this out:
  \url{https://eccc.weizmann.ac.il/report/2018/159/}. Quickly checking
  the parameters I believe it does not give anything better. Need to
  check more carefully}

For our purposes it is more relevant to compute the neighbor relation
$\neigh(u,v)$ indicating whether $(u, v) \in E$ rather than the
neighbor function as in \cref{def:explicit}, but this is an immediate
consequence of being able to compute the neighbor function.

\begin{claim}
  \label{claim:explicit neigh}
  If $G = (U, V, E)$ is explicit then the neighbor relation
  $\neigh: U \times V \rightarrow \{0,1\}$ is computable by a circuit
  of size
  $d \cdot \big( \poly(\log n + \log |U| + \log d) + 2 \log |V| + 1
  \big)$.
\end{claim}

% \pa{Added this fact, check that I was not too optimistic in claimed
  % size bound. Maybe add proof here.} \kr{checked.}

% Recall that this graph is explicit and hence there are
% circuits to compute the neighbor function
% $\Gamma: U \times [k] \rightarrow Z$.

% Each circuit $C_1, \ldots, C_z$ consists of $k$ circuits $\calD =
% \set{D_1, \ldots, D_k}$ that compute $\Gamma$ with the second input
% fixed to their corresponding index. On top of these circuits from
% $\calD$ we add a small circuit  $E_i$, which computes the following.
% Each circuit $D \in \calD$ outputs
% some index $j \in [z]$. The circuit $E_i$ outputs $1$ if and only if
% one of the $D \in \calD$ outputs the index $i$. The circuit $E_i$ is
% fairly small: for each circuit $D$ we require $2 \log z$ gates to
% check whether $D$ output $i$, and an additional $k$ gates to \emph{or}
% these outputs together. So in total each $E_i$ consists of
% $k (2\log z + 1)$ many gates.
% As the graph $G$ is explicit, we know that there is a circuit of size
% $\poly(n + \log k)$ computing $\Gamma$ and hence the size of all
% circuits in $\calD$ is bounded by $c_\alpha \cdot k \cdot n^c$, for
% some constant $c$, and $c_\alpha$ depending on $\alpha$. Hence we
% conclude that
% \begin{align}\label{eq:bound-a}
%  a = z \cdot \big(k (2\log z + 1) + c_\alpha \cdot k \cdot n^c\big) \eqperiod
% \end{align}

A slice function is a Boolean function $f$ such that there is a
$\ell \in [n]$ with $f(\alpha) = 0$ whenever $|\alpha| < \ell$, and
$f(\alpha) = 1$ whenever $|\alpha| > \ell$.  Note that all slice
functions are monotone.

The circuit complexity $\csize(f)$ of a Boolean function $f$ is the
size of the smallest circuit over the basis $\vee, \wedge$, and
$\lnot$ (with fan-in $2$). Similarly the monotone circuit
complexity $\cmonsize(f)$ of a monotone Boolean function $f$
is the size of the smallest circuit over the basis $\vee$, and
$\wedge$. We have the following useful inequality between these
measures.

\begin{lemma}[\cite{B82}]\label{lem:slice}
  If $g$ is any slice function on $n$ bits, then
  $\cmonsize(g) \le 2 \csize(g) + O(n^2 \log n)$.
\end{lemma}

Finally we also rely on the following simple claim.

\begin{claim}
  \label{claim:0 mod bool}
  Let $p: \R^n \rightarrow \R$ be a degree $d$ polynomial such that
  $p(x) = 0$ for all $x \in \{0,1\}^n$.  Then $p$ can be written as
  \[
    p(x) = \sum_{i \in [n]} q_i(x) \cdot (x_i^2-x_i)
  \]
  where each term in the sum has degree at most $d$.
\end{claim}

\begin{proof}[Proof sketch]
  We take the polynomial $p$ and multilinearize it, using the
  appropriate polynomial $x_i^2 - x_i$. Eventually we are left with a
  sum of polynomials of the form $q_i(x)\cdot(x_i^2 - x_i)$ and a
  multilinear polynomial $\tilde p(x)$ which is $0$ on all Boolean
  inputs.  As multilinear polynomials are a basis for Boolean
  functions this implies that $\tilde p(x)$ is equal to the $0$
  polynomial and hence the claim follows.
\end{proof}

\subsection{Sum of Squares}
\label{sec:known-SoS-results}

Let $\calP = \set{p_1 = 0, \ldots, p_m = 0}$ be a system of polynomial
equations over the set of variables
$X = \set{x_1, \ldots, x_n, \bar{x}_1, \ldots, \bar{x}_n}$.  Each
$p_i$ is called an \emph{axiom}, and throughout the paper we always
assume that $\calP$ includes all axioms $x_i^2 - x_i$ and
$\bar{x}_i^2 - \bar{x}_i$, ensuring that the variables are Boolean, as
well as the axioms $1 - x_i - \bar{x}_i$, making sure that the ``bar''
variables are in fact the negation of the ``non-bar'' variables.

\begin{definition}[Sum-of-Squares]
  \emph{Sum-of-Squares (SoS)} is a static, semi-algebraic proof
  system. An SoS proof of $f \ge 0$ from $\calP$ is a sequence of
  polynomials $\pi = (t_1, \ldots, t_m; s_1, \ldots, s_a)$ such that
  \begin{align*}
    \sum_{i \in [m]} t_i p_i +
    \sum_{i \in [a]} s^2_i = f \eqperiod %\label{eq:sos-def}
  \end{align*}
  The \emph{degree} of a proof $\pi$ is
  $\Deg(\pi) = \max \set{ \max_{i \in [m]} \deg(t_i) + \deg(p_i),
    \max_{i \in [a]} 2 \deg(s_i) }$, an \emph{SoS refutation of
    $\calP$} is an SoS proof of $-1 \ge 0$ from $\calP$, and the
  \emph{SoS degree to refute $\calP$} is the minimum degree of any SoS
  refutation of $\calP$: if we let $\pi$ range over all SoS
  refutations of $\calP$, we can write
  $\RefuteDeg{SoS}{\calP} = \min_\pi \Deg(\pi)$.  The \emph{size of an
    SoS refutation $\pi$}, $\Size(\pi)$, is the sum of the number of
  monomials in each polynomial in $\pi$ and the \emph{size of refuting
    $\calP$} is the minimum size over all refutations
  $\RefuteSize{SoS}{\calP} = \min_{\pi} \Size(\pi)$.
\end{definition}

Let us recall some well-known results about SoS. Given a bipartite
graph $G = (U, V, E)$, and $b \in \set{0,1}^{|U|}$ we denote by
$\Phi(G,b)$ the following XOR-CSP instance defined over $G$: for each
$v \in V$ there is a Boolean variable $x_v$, and for every vertex
$u \in U$ there is a constraint $\oplus_{v \in N(u)} x_v = b_u$.
We encode this in the obvious way as a system of polynomial equations: 
\begin{align*}
  \big\{
  \prod_{v \in N(u)} (1 - 2 \cdot x_v) = 1 - 2 \cdot b_u
  \mid u \in U
  \big\} \eqcomma
\end{align*}
along with the Boolean axioms and the negation axioms for the $x$
variables. The first theorem we need to recall is the classic lower
bounds for XOR-CSPs by Grigoriev.

\begin{theorem}[\cite{gri01xor}]\label{thm:sos-xor}
  For $n \in \N$, all $k = k(n)$ and $r = r(n)$ the following holds.
  Let $G = (U, V, E)$ be an $(r, k, 2)$-expander with $|V| = n$. Then
  for every $b \in \set{0,1}^{|U|}$ SoS requires degree $\Omega(r)$ to
  refute the claim that there is a satisfying assignment to
  $\Phi(G, b)$.
\end{theorem}

We also need to recall the size-degree tradeoff by Atserias and
Hakoniemi.

\begin{theorem}[\cite{AH18SosTradeoff}]\label{thm:tradeoff}
  Let $\calP$ be a system of polynomial equations over $n$ Boolean
  variables and degree at most $k$. If $d$ is the minimum degree
  SoS requires to refute $\calP$, then the minimum size of an SoS
  refutation of $\calP$ is at least $\exp(\Omega((d-k)^2/n))$.
\end{theorem}

\subsection{Restrictions}
\label{sec:restrictions}

Let $\calP = \set{p_1 = 0, \ldots, p_m = 0}$ be a system of polynomial
equations over the set of Boolean variables
$X = \set{x_1, \ldots, x_n, \bar{x}_1, \ldots, \bar{x}_n}$.
For a map
$\rho: \set{x_1, \ldots, x_n} \to \set{0,1, x_1, \ldots, x_n,
  \bar{x}_1, \ldots, \bar{x}_n}$ denote by $\restrict{\calP}{\rho}$
the system of polynomial equations $\calP$ restricted by $\rho$, i.e.,
\begin{align*}
\restrict{\calP}{\rho} = \set{ &p_1(\rho(x_1), \ldots, \rho(x_n)) = 0,\\
  &p_2(\rho(x_1), \ldots, \rho(x_n)) = 0,\\ &\vdots\\ &p_m(\rho(x_1),
  \ldots, \rho(x_n)) = 0 }\eqcomma
\end{align*}
where it is understood that $\rho(\bar x_i) = \overline{\rho(x_i)}$,
with the convention $\bar{\bar{x}}_i = x_i$, $\bar 0 = 1$ and vice
versa. Throughout the paper all our restrictions set the bar variables
to the negation of the non-bar variables. As such it makes sense to
treat the pair of variables $(x_i, \bar x_i)$ as one variable and we
say that $\calP$ has $n$ \emph{unset} variables.

\begin{definition}[Variable Substitution]
  We say that a system of polynomial equations $\calP'$ is a
  \emph{variable substitution of $\calP$} if there is a map
  $\rho: \set{x_1, \ldots, x_n} \to \set{0,1, x_1, \ldots, x_n,
    \bar{x}_1, \ldots, \bar{x}_n}$ such that
  $\calP' = \restrict{\calP}{\rho}$, where we ignore polynomial
  equations of the form $0=0$.
\end{definition}

The following well-known lemma states that a system of polynomial
equations $\calP$ is at least as hard as any of its variable substitutions.
\begin{lemma}\label{lem:subformula}
  Let $\calP, \calP'$ be systems of polynomial equations such that
  $\calP' \subformula \calP$. Then,
  \begin{enumerate}[label=$(\roman*)$]
  \item $\RefuteDeg{SoS}{\calP} \ge \RefuteDeg{SoS}{\calP'}$, and
  \item $\RefuteSize{SoS}{\calP} \ge \RefuteSize{SoS}{\calP'}$.
  \end{enumerate}
\end{lemma}
The lemma is easy to verify by considering an SoS refutation of
$\calP$ and hitting it with the appropriate variable substitution. The
restricted proof is now a refutation of $\calP'$ and it can be seen
that the degree/size of the restricted refutation is at most the
degree/size of the original refutation.

We also consider more general substitutions.

\begin{definition}[Polynomial Substitution]
  Functions
  $\rho: \set{x_1, \ldots, x_n} \rightarrow \R[x]_{\le k}$ that map
  variables to polynomials of degree at most $k$ are called
  \emph{polynomial substitutions}.
\end{definition}

For polynomial substitutions we have the following well-known lemma.

\begin{lemma}
  \label{lem:sos degree polynomial substitution}
  Let $\calP$ be a system of polynomial equations and let $\rho$ be a
  polynomial substitution mapping variables to polynomials of degree
  at most $k$. Then,
  $\RefuteDeg{SoS}{\calP} \ge
  \RefuteDeg{SoS}{\restrict{\calP}{\rho}}/k$.
\end{lemma}

This lemma can again be verified by considering a refutation of
$\calP$. Substitute each variable $x_i$ in the proof by
$\rho(x_i)$. This results in a refutation of $\restrict{\calP}{\rho}$,
whose degree is at most a factor $k$ larger than the degree of the
refutation of $\calP$.

\subsection{The Circuit Size Formula}
\label{sec:circuit-encoding}
The formula $\circuit_s(f)$ encodes the claim that the function $f$,
given as a truthtable $f \in \set{0, 1}^{2^n}$, can be computed by a
circuit of size $s$ over $n$ Boolean inputs $x_1, \ldots, x_n$. The
encoding is not essential but for concreteness let us fix one encoding
of this claim. We deviate from the encoding used by Razborov
\cite{Razborov98, Razborov04ResolutionLowerBoundsPM} and do not
present the formula as a propositional formula but rather as a system
of polynomial equations. In order to encode below constraints as a
constant width CNF formula, as done by Razborov, one needs to
introduce extension variables. Despite this difference it is not
difficult to see that our lower bound also works against the CNF
encoding. In \cref{sec:encoding} we directly show that a low degree
SoS refutation of the CNF encoding gives rise to a low degree SoS
refutation of the encoding used in this paper (see
\cref{prop:cnf-poly}). Thus a lower bound for our encoding implies a
lower bound for the CNF encoding. As the presentation is simpler in
the polynomial encoding, we present it as follows.

We also need to define the monotone version of $\circuit_s(f)$ denoted
by $\circuit^{\mathsf{mon}}_s(f)$. The later is a restriction of the
former with the $\isneg(v)$ (see below) variable, for all $v \in [s]$,
set to $0$. This forces the circuit to only contain $\land$ and $\lor$
gates, i.e., the circuit is monotone.

All variables introduced in the following are Boolean variables and we
implicitly add the Boolean axiom $y(1 - y) = 0$ for each variable $y$
and further implicitly introduce the ``bar variable'' $\bar{y}$ along
with the negation axiom $y = 1 - \bar{y}$ (and the corresponding
Boolean axiom) ensuring that $\bar{y}$ is always the negation of $y$.

Let us first describe the \emph{structure variables} which are used to
describe the circuit that supposedly computes the function $f$.

We view the $s$ gates as being indexed from $1$ to $s$ in topological
order with gate $s$ being the output.
For each gate $v \in [s]$ there are three variables $\isneg(v),
\isor(v), \isand(v)$ indicating the operation computed at
$v$. Similarly for a gate $v \in [s]$ and a wire $a \in \set{1,2}$ we
have variables $\isfromconst(v, a), \isfromvar(v, a), \isfromgate(v, a)$
indicating whether the input wire $a$ of $v$ is connected to a
constant, a variable or a gate.

Further, we have the variables $\constval(v, a)$, $\isvar(v, a, i)$
and $\isgate(v, a, u)$, for $a \in \set{1,2}$, $i \in [n]$ and
$u < v$, specifying the constant value, input $x_i$ or gate $u$, the
input wire $a$ of $v$ is connected to (assuming $a$ is connected to
the corresponding kind).

The second set of variables are the \emph{evaluation variables}, which
describe what value is computed at each $v$ on input
$\alpha = \alpha_1, \ldots, \alpha_n$ (i.e., we have $x_i = \alpha_i$).

For each gate $v \in [s]$ and assignment $\alpha \in \set{0,1}^n$ we
have a Boolean variable $\outwire_\alpha(v)$ indicating the value
computed at gate $v$ on input $\alpha$. The Boolean variable
$\inwire_{\alpha}(v, a)$ indicates the value brought to the vertex
$v \in [s]$ on wire $a \in \set{1,2}$ on input $\alpha$.

Note that there is a total of
$3s + 6s + 2s + 2sn + 2{s \choose 2} = \Theta(s^2+sn)$ structure
variables, and a total of $3 s 2^n$ evaluation variables, for a total
of $\Theta(s^2 + s 2^n)$ variables in $\circuit_s(f)$.

The formula consists of the following axioms. For the sake of
readability we omit some universal quantifiers: the variable
$a \in \set{1,2}$ in
\cref{eq:ax-from,eq:ax-in-gt,eq:ax-in-gt2,eq:prod-zero-in,eq:prod-zero-gate,eq:inwire-const,eq:inwire-var,eq:inwire-gate}
as well as the variable $\alpha \in \set{0,1}^n$ in
\cref{eq:inwire-const,eq:inwire-var,eq:inwire-gate,eq:ax-neg,eq:ax-or,eq:ax-and,eq:ax-correct}
are implicitly universally quantified.

Let us first describe
the axioms on the structure of the circuit. In the following
section we refer to this set of axioms as the \emph{structure
  axioms}. The first axioms ensure that every wire is connected to a
single kind
\begin{align}
  \label[axiom]{eq:ax-from}
  \isfromconst(v, a) +
  \isfromvar(v, a) +
  \isfromgate(v, a) = 1 \quad
  \forall \, v \in [s] \eqcomma
\end{align}
and similarly the next axioms make sure that each gate is of
precisely one kind
\begin{align}
  \label[axiom]{eq:ax-fn}
  \isneg(v) + \isor(v) + \isand(v) = 1
  \quad \forall \, v \in [s] \eqperiod
\end{align}
The final structure axioms ensure that the variables, which indicate
to what input or gate a fixed wire is connected to, always sum to one
(except for gate $1$ which cannot have any inputs from other gates)
\begin{align}
  \label[axiom]{eq:ax-in-gt}
  \sum_{i = 1}^{n}\isvar(v,a,i) &= 1
  \quad \forall v \in [s], \text{ and}\\
  \label[axiom]{eq:ax-in-gt2}
  \sum_{u = 1}^{v-1}\isgate(v,a,u) &= 1
  \quad \forall v \in [s] \setminus\{1\} \eqperiod
\end{align}
We further strengthen our encoding by adding the axioms
\begin{align}
  \label[axiom]{eq:prod-zero-in}
  \isvar(v,a,i)\isvar(v,a,j) &= 0
  \quad\forall v \in [s],\, i < j \in [n], \text{ and}\\
  \label[axiom]{eq:prod-zero-gate}
  \isgate(v,a,u)\isgate(v,a,u') &= 0
  \quad\forall u < u' < v \in [s] \eqperiod
\end{align}
Note that \cref{eq:prod-zero-in,eq:prod-zero-gate} are implied by
\cref{eq:ax-in-gt,eq:ax-in-gt2}. We add these axioms in order to argue
that a short refutation of the CNF encoding of this principle leads to
a short refutation of the present encoding.

The second group of axioms are the \emph{evaluation axioms} and they
ensure that the evaluation variables indeed compute the intended
values. We start by making sure that the wires carry the value
intended by the structure axioms. If a wire is connected to a
constant, then the evaluation variable associated with that wire
should always be equal to the constant
\begin{align}
  \label[axiom]{eq:inwire-const}
  \isfromconst(v,a) \cdot
  \big( \inwire_{\alpha}(v,a) - \constval(v, a) \big) = 0 \eqcomma
\end{align}
and similarly in case if a wire is connected to an input or a gate
\begin{align}
  \label[axiom]{eq:inwire-var}
  \isfromvar(v, a) \cdot \isvar(v, a, i) \cdot
  \big( \inwire_\alpha(v, a) - \alpha_i\big) &= 0
  \eqcomma\\
  \label[axiom]{eq:inwire-gate}
  \isfromgate(v, a) \cdot \isgate(v, a, u) \cdot
  \big(\inwire_\alpha(v, a) - \outwire_\alpha(u)\big) &= 0 \eqperiod
\end{align}
The final set of evaluation axioms makes sure that the output
evaluation variable of a gate is correctly related to the input
evaluation variables:
\begin{align}
  \label[axiom]{eq:ax-neg}
  \isneg(v) \cdot
  \outwire_{\alpha}(v)
  &=
    \isneg(v) \cdot
    \overline{\inwire_{\alpha}(v, 1)}\eqcomma\\
  \label[axiom]{eq:ax-or}
  \isor(v) \cdot
  \outwire_{\alpha}(v)
  &=
    \isor(v)
    \cdot \big(
    1 -
    \overline{\inwire_{\alpha}(v, 1)} \cdot
    \overline{\inwire_{\alpha}(v, 2)}
    \big) \eqcomma\\
  \label[axiom]{eq:ax-and}
  \isand(v) \cdot
  \outwire_{\alpha}(v)
  &=
    \isand(v) \cdot
    \inwire_{\alpha}(v, 1) \cdot
    \inwire_{\alpha}(v, 2) \eqperiod
\end{align}

Last but not least we have the axioms that ensure that the circuit
outputs the function specified by the truthtable
\begin{align}
  \label[axiom]{eq:ax-correct}
  \outwire_{\alpha}(s) &= f(\alpha) \eqperiod
\end{align}

%% file: circuits-and-restrictions.tex
\section{On Circuits and Restrictions}
\label{sec:circ-rest}

Let $G=(U, V, E)$ be a bipartite graph with $U = \set{0,1}^n$ and
$V = [m]$. As in the XOR-CSP setup (\cref{sec:known-SoS-results}) we
think of vertices in $U$ as constraints and vertices in $V$ as
variables. More specifically, we think of each vertex $\alpha \in U$
as an \emph{xor} constraint over the variables in the neighborhood
$\oplus_{i \in N(\alpha)} v_i = b_\alpha$, for a constraint vector
$b \in \set{0,1}^U$. Given an assignment $\beta \in \set{0,1}^m$ to
the variables $V$, we let $f_{G, \beta}: U \rightarrow \set{0,1}$ be
the function defined by
$f_{G, \beta}(\alpha) = \oplus_{i \in N(\alpha)} v_i$.  In other
words, viewing $f_{G,\beta}$ as a vector in $\set{0,1}^U$, it is the
unique constraint vector such that the XOR-CSP instance, defined over
$G$, is satisfied by the assignment $\beta$. Let us denote the set of
all such constraint vectors that give rise to a satisfiable XOR-CSP
instance by
\[
  \calF_G = \set{f_{G, \beta} \,|\, \beta \in \set{0,1}^m} \eqperiod
\]
In order for SoS to refute an XOR-CSP instance defined over $G$,
it must prove that the given constraint vector is not in the set
$\calF_G$.

On the other hand in order for SoS to refute the formula $\circuit_s(f)$
it needs to show that there is no circuit of size at most $s$
computing $f$. That is, SoS needs to show that $f$ is not in the set
\[
  \calC_{\emptyset} =
  \set{
    T:
    \set{0,1}^n \rightarrow \set{0,1}
    \text{~such that~}
    \circuit_s(T)
    \text{~is satisfiable}
  } \eqperiod
\]
More generally, if we restrict $\circuit_s(f)$ by a restriction
$\rho$, then the proof system must prove that $f$ is not a member
of the family of truthtables
\[
  \calC_\rho =
  \set{
    T:
    \set{0,1}^n \rightarrow \set{0,1}
    \text{~such that~}
    \restrict{\circuit_s(T)}{\rho}
    \text{~is satisfiable}
  } \eqperiod
\]

In the following we show that there is a well-behaved restriction
$\rho$ such that $\calC_\rho = \calF_G$ for some explicit graphs
$G$. In other words, once we consider the restricted formula
$\restrict{\circuit_s(f)}{\rho}$, SoS needs to rule
out that $f$ is a valid right hand side of an XOR-CSP instance. But we
know that if $G$ is a moderate expander, then low degree SoS cannot
determine whether the XOR-CSP instance is satisfiable and hence we
obtain our lower bound.

Let us first formalize the properties we require from $\rho$. We start
off by restricting our attention to a certain natural class of variable
substitutions. Namely, we do not want that the structure of the circuit
depends on evaluation variables.

\begin{definition}[natural variable substitutions]
  A variable substitution $\rho$ to the variables of $\circuit_s(f)$ is
  \emph{natural} if there is \emph{no} structure variable $y$ such
  that $\rho(y)$ is an evaluation variable.
\end{definition}

In order to motivate the next definition, let us informally
describe the natural restriction $\rho$ and explain the properties of
$\rho$ we require.

For now we can think of $\rho$ as a restriction to the structure
variables (though for the size lower bounds we also need to restrict
some of the evaluation variables). Some set of $m$ structure variables remains
undetermined. Let us denote these variables by $y_1, \ldots, y_m$. We
intend to choose $\rho$ such that on a given input
$\alpha \in \set{0,1}^n$ to the circuit, it is forced to compute
$\oplus_{i \in N(\alpha)} y_i$.  In other words, given such a
restriction $\rho$, we are \emph{essentially} left with an XOR-CSP
problem over $G$, with right hand side $f$.  There is however a
difference in that the encoding is non-standard: the evaluation
variables act like extension variables that correspond to the
functions computed at each gate of the circuit. In order to argue that
the known degree lower bound for the XOR-CSP problem implies a degree
lower bound for the problem at hand, we need to get rid of these
extension variables. This can be done if the functions computed at the
gates are of low degree in the $y$ variables.

Recall from \cref{sec:restrictions} that a system of polynomial
equations $\calP$ has $n$ unset variables if there are $n$ tuples of
variables $(x, \bar x)$ such that at least one variable of each tuple
occurs in $\calP$ and all variables in these tuples are unset, i.e.,
they are not fixed to a constant.

\begin{definition}[$k$-determined]
  \label{def:determined}%
  Let $\rho$ be a variable substitution to the variables of
  $\circuit_s(f)$ and suppose that $\rho$ leaves $m$ structural
  variables $Y = \set{y_1, \ldots, y_m}$ unset.  Then $\rho$ is
  \emph{$k$-determined} if for every $v \in [s]$ and
  $\alpha \in \set{0,1}^n$ there are multilinear polynomials
  $$\gout_{v, \alpha}, \ginone_{v, \alpha}, \gintwo_{v, \alpha}:
  \set{0,1}^m \rightarrow \set{0,1}$$ depending on at most $k$
  variables such that the following holds. For all $T \in \calC_\rho$
  and all total assignments $\sigma$ that satisfy
  $\restrict{\circuit_s(T)}{\rho}$ it holds that
  \begin{align}
    \label{eq:k-determined-substitution}
    \restrict{\outwire_\alpha(v)}{\rho \cup \sigma} &=
    \gout_{v, \alpha}(\beta) \eqcomma &
    \restrict{\inwire_\alpha(v, 1)}{\rho \cup \sigma} &= \ginone_{v,
    \alpha}(\beta) \eqcomma \text{~and} &
    \restrict{\inwire_\alpha(v, 2)}{\rho \cup \sigma} &= \gintwo_{v,
    \alpha}(\beta) \eqcomma
  \end{align}
  where $\beta \subseteq \sigma$ is the assignment to $Y$.
\end{definition}

However, \cref{def:determined} is not quite sufficient. For example,
there is no guarantee that $\calC_\rho$ is non-empty, i.e., that the
restriction $\rho$ describes a valid (partial) circuit. More generally, we
need the additional guarantee that there are still many viable
circuits that the restricted formula can describe: if there is just a
single setting of the $Y$ variables such that all structural axioms
are satisfied, then the formula may be refuted in constant
degree. Hence we need to ensure that there are many viable assignments
to the $Y$ variables that satisfy all structure axioms. This leads us
to the following definition.

\begin{definition}[$m$-independent]
  A variable substitution $\rho$ to the variables of the formula
  $\circuit_s(f)$ is \emph{$m$-independent} if $\rho$ leaves exactly
  $m$ structural variables $Y = \set{y_1, \ldots, y_m}$ unset, and for
  every assignment $\beta \in \set{0,1}^Y$ it holds that
  $|\calC_{\rho \cup \beta}| = 1$.
\end{definition}

With these definitions at hand we can state the lemma that drives all
our lower bounds.

\begin{lemma}
  \label{lem:circuit-to-poly}
  Let $\rho$ be a natural $m$-independent $k$-determined variable substitution
  of $\circuit_s(f)$, and let $Y$ and
  $\gout_{u, \alpha}$ be as in \cref{def:determined}. If there is an
  SoS refutation of $\restrict{\circuit_s(f)}{\rho}$ of degree $d$,
  then there is a degree $d \cdot k$ SoS refutation of the system
  of polynomial equations
  \begin{align}
    \label{eqn:lb transfer lemma}
    \set{
      \gout_{s,\alpha}(Y) = f(\alpha) \mid \alpha \in \set{0,1}^n
    }
    \cup
    \set{
      y_i^2 = y_i \mid i \in [m]
    } \eqperiod
  \end{align}
\end{lemma}

For proving this lemma, we consider the natural extension of $\rho$
which substitutes all evaluation variables by appropriate
degree-$k$ polynomials as indicated by \cref{def:determined}.

\begin{definition}
  For a $k$-determined restriction $\rho$ (with associated polynomials
  $\gout_{v,\alpha}, \ginone_{v,\alpha}, \gintwo_{v,\alpha}$) of
  $\circuit_s(f)$, we denote by $\hat\rho$ the polynomial substitution
  that extends $\rho$ by first substituting any bar variable $\bar x$
  by $1 - x$ and then substituting all evaluation variables as
  follows:
  \begin{align*}
    \hat\rho(\outwire_\alpha(v)) &=
    \gout_{v, \alpha}(Y) \eqcomma &
    \hat\rho(\inwire_\alpha(v, 1)) &= \ginone_{v,
    \alpha}(Y) \eqcomma \text{~and} &
    \hat\rho(\inwire_\alpha(v, 2)) &= \gintwo_{v,
    \alpha}(Y) \eqperiod
  \end{align*}
\end{definition}

Note that the formula $\restrict{\circuit_s(f)}{\hat\rho}$ is
defined only over $Y$. Let us stress that there are no ``bar''
variables left in the formula. The main observation used to prove
\cref{lem:circuit-to-poly} is the following claim, which establishes
that the formula \eqref{eqn:lb transfer lemma} is in fact essentially
the same as $\restrict{\circuit_s(f)}{\hat\rho}$.

\begin{claim}
  \label{claim:tau-rho}
  Let $\rho$ be a natural $m$-independent $k$-determined variable substitution
  of $\circuit_s(f)$.  Then
  $\restrict{\circuit_s(f)}{\hat\rho}$ can be written as
  $$
  \restrict{\circuit_s(f)}{\hat\rho} = \calP \cup \calQ \eqcomma
  $$
  where $\calP$ is the formula \eqref{eqn:lb transfer lemma} and
  $\calQ$ only consists of axioms that are satisfied for all
  assignments $\beta \in \{0,1\}^Y$.
\end{claim}

\begin{proof}
  Note that the set of output axioms \eqref{eq:ax-correct} of
  $\circuit_s(f)$ under $\hat\rho$ equals the first part of
  \eqref{eqn:lb transfer lemma}, and that the Boolean axioms on the $Y$ variables in
  $\restrict{\circuit_s(f)}{\hat\rho}$ are exactly the second part
  of \eqref{eqn:lb transfer lemma}.

  The remaining axioms of $\restrict{\circuit_s(f)}{\hat\rho}$, which
  are not present in \eqref{eqn:lb transfer lemma}, are the Boolean
  axioms on the variables outside $Y$, the negation axioms, as well as
  \cref{eq:ax-from,eq:ax-fn,eq:ax-in-gt,eq:ax-in-gt2,eq:prod-zero-in,eq:prod-zero-gate,eq:inwire-const,eq:inwire-var,eq:inwire-gate,eq:ax-neg,eq:ax-or,eq:ax-and}.

  The Boolean axioms may turn into polynomials of degree at most
  $2k$. Because the polynomials we substitute the variables with are
  Boolean valued, we see that these substituted axioms are satisfied
  for all assignments $\beta \in \set{0,1}^Y$ and we can thus put them
  into the set $\calQ$.

  The negation axioms all become ``$0=0$'' under $\hat{\rho}$ since
  $\hat{\rho}(\bar{x}) = 1 - \hat{\rho}(x)$.
  
  Finally we need to argue that the
  \cref{eq:ax-from,eq:ax-fn,eq:ax-in-gt,eq:ax-in-gt2,eq:prod-zero-in,eq:prod-zero-gate,eq:inwire-const,eq:inwire-var,eq:inwire-gate,eq:ax-neg,eq:ax-or,eq:ax-and}
  are also of the form $p(Y)=0$ for a polynomial $p$ which is
  identically $0$ on all of $\{0,1\}^m$.  This in turn follows
  immediately from the assumption that $\rho$ is $m$-independent: for
  every $\beta \in \{0,1\}^Y$, there exists some $T$ such that the
  complete assignment $\hat\rho(\beta) \cup \beta$ satisfies
  $\circuit_s(T)$.  But since none of the remaining
  \cref{eq:ax-from,eq:ax-fn,eq:ax-in-gt,eq:ax-in-gt2,eq:prod-zero-in,eq:prod-zero-gate,eq:inwire-const,eq:inwire-var,eq:inwire-gate,eq:ax-neg,eq:ax-or,eq:ax-and}
  depends on $T$, they must then all be satisfied for every
  $\beta \in \{0,1\}^Y$.
\end{proof}

Using this claim we can easily prove \cref{lem:circuit-to-poly}.

\begin{proof}[Proof of \cref{lem:circuit-to-poly}]
  Suppose $\restrict{\circuit_s(f)}{\rho}$ has a refutation in degree
  $d$.  By \cref{lem:sos degree polynomial substitution}, there then
  exists a degree $d \cdot k$ refutation of
  $\restrict{\circuit_s(f)}{\hat\rho}$.

  By \cref{claim:tau-rho}, this new refutation is \emph{almost} a
  refutation of \eqref{eqn:lb transfer lemma}, except that
  $\restrict{\circuit_s(f)}{\hat\rho}$ has an additional set $\calQ$
  of axioms that the refutation may use.  However, each of these
  additional axioms is of the form $p(Y) = 0$ for a polynomial which
  is identically $0$ on the entire Boolean cube.  By \cref{claim:0 mod
    bool}, such an axiom can be rewritten as a linear combination of
  the Boolean axioms.  Since the Boolean axioms are present in
  \eqref{eqn:lb transfer lemma}, this yields a refutation of that
  formula in degree $d \cdot k$.
\end{proof}

%% file: de-morgan-lb.tex
\section{Lower Bounds for General Circuits}
\label{sec:demorgan-LB}

We state the following lemma general enough so that we can apply it
for the degree as well as the size lower bound. As explained
previously, for the size lower bounds we rely on functions that almost
have circuits of size $s$. Recall that we consider the class of
functions $\calF_n(s,t)$ that consists of all Boolean functions
$f: \set{0,1}^n \rightarrow \set{0,1}$ for which there is a Boolean
circuit $C_f: \set{0,1}^n \rightarrow \set{0, 1, \bot}$ of size at most
$s$ such that
\begin{enumerate}
\item if $C_f(\alpha) \neq \bot$, then $C_f(\alpha) = f(\alpha)$, and
\item $C_f(\alpha) = \bot$ on at most $t$ inputs.
\end{enumerate}
The following lemma establishes the existence of $m$-independent
$k$-deter\-mined variable substitutions that result in XOR-CSP instances
over explicit graphs.

\begin{lemma}
  \label{lem:restriction}
  For all $k, m, n, t \in \N$ satisfying $m \le 2^n$, and any explicit
  bipartite graph $G = (U, V, E)$ such that $|U| = 2^n$, $|V| = m$ and
  all $u \in U$ are of degree $\deg(u) \le k$, the following holds.
  There is a constant $C > 0$, depending on the explicitness of $G$,
  such that for all $s \ge C \cdot m \cdot n^C\cdot k^C$ and any
  Boolean function $f \in \calF_n(s/2,t)$ there is a natural
  $m$-independent $k$-determined variable substitutions $\rho$ for the
  formula $\circuit_s(f)$ such that
  \begin{align*}
    \gout_{s, \alpha}(Y) =
    \begin{cases}
      f(\alpha), &\text{if~} C_f(\alpha) \neq \bot,\\
      \oplus_{i \in N(\alpha)} y_i, &\text{otherwise}
    \end{cases}
  \end{align*}
  for all $\alpha \in \set{0,1}^n$ and $\gout_{s, \alpha}$ and $Y$
  as in \cref{def:determined}.  Furthermore, the formula $\restrict{\circuit_s(f)}{\rho}$ is over
  $O\big(t \cdot k + m\big)$ variables.
\end{lemma}

% \pa{Issue with how explicit is defined (see preliminaries) makes this
%   a bit loosey-goosey.  The exact size of the selector circuits
%   depends on $G$ and this governs the hidden constant in the $\Omega$
%   for how large $s$ needs to be.  In particular it is not ``for
%   every'' $s = \Omega(...)$ (which indicates that the hidden constant
%   can be made as small as you want).}
% \kr{Better? At least the statement?}

For the degree lower bound (\cref{thm:main-degree}) we will set $t =
2^n$ and use the trivial $C_f$ which always outputs $\bot$, so the
reader who wishes a simplified version of the lemma can focus on this
special case.

\begin{proof}
  We consider the formula $\circuit_s(f)$ and let the first $m$ gates
  of the formula be denoted by $Y$. We restrict the formula such that
  each gate in $Y$ computes an \emph{or} of two constants. The first
  wire to the gate is fixed to the constant $0$, whereas the second
  wire is only restricted to carry either the constant $0$ or $1$.  In
  the end these will be the only structural variables that are not
  restricted to a constant. In the following we think of the gates $Y$
  as Boolean variables; as $m$ additional input bits to our circuit.

  Further, we restrict another part of the formula such that one part
  of the circuit described by the formula computes the circuit
  $C_f$. Recall that we pretend that the output of $C_f$ is in
  $\set{0,1,\bot}$, but it actually outputs two bits $C_f^1$ and
  $C_f^2$, where $C_f^1(\alpha) = 1$ if and only
  if $C_f^2(\alpha) = f(\alpha)$.

  Finally we also want to hard code the bipartite graph
  $G(\set{0,1}^n, Y, E)$ into our circuit. Since $G$ is very large
  this requires $G$ to be explicit. That is, we require small circuits
  $\sel_1, \ldots, \sel_{m}$, where given any $\alpha \in \set{0,1}^n$, $\sel_i(\alpha)$ is $1$ if and only if the vertex $y_i \in Y$ is a
  neighbor of the vertex $\alpha$.
  By \cref{claim:explicit neigh} these circuits $\sel_i$ are each of size
  \[
    k \cdot \left( \poly(n + \log k) + 2 \log m + 1 \right)
    \le \poly(n, k) \eqperiod
  \]
  The restriction $\rho$ restricts some structural variables such that
  a part of the circuit computes $\sel_1, \ldots, \sel_m$. We connect
  each output of the $\sel_i$ circuit by an \emph{and} gate to the
  negation of $C_f^1$. Denote the resulting circuits by
  $\sel'_1, \ldots, \sel'_m$.  Observe that the circuits $\sel_i'$
  output $0$ whenever $C_f^2(\alpha) = f(\alpha)$ and otherwise output
  $\sel_i$. We think of these circuits as ``selector circuits'' which
  indicate whether on input $\alpha \in \set{0,1}^n$ (to the original
  variables $x_1,\ldots, x_n$ over which the circuit is defined) the
  variable $y_i \in Y$ appears in the constraint for $\alpha$.

  The output of these selector circuits $\sel_i'$ is connected to the
  gate $y_i$ by an \emph{and} gate. All these $m$ \emph{and} gates are
  in turn connected to a circuit computing the \emph{xor} of these
  gates. Finally, to ensure that the circuit computes $f(\alpha)$ on
  inputs $\alpha$ such that $C_f(\alpha) \neq \bot$, we connect
  $C_f^1$ with $C_f^2$ by an \emph{and} gate which is then connceted
  by a \emph{or} gate to the output of the \emph{xor} circuit. This
  completes the description of the restriction on the structure
  variables. A depiction of the resulting circuit can be found in
  \cref{fig:circuit-restricted}.
  \begin{figure}
    \centering
    \includegraphics[width=\linewidth,height=\textheight,keepaspectratio]{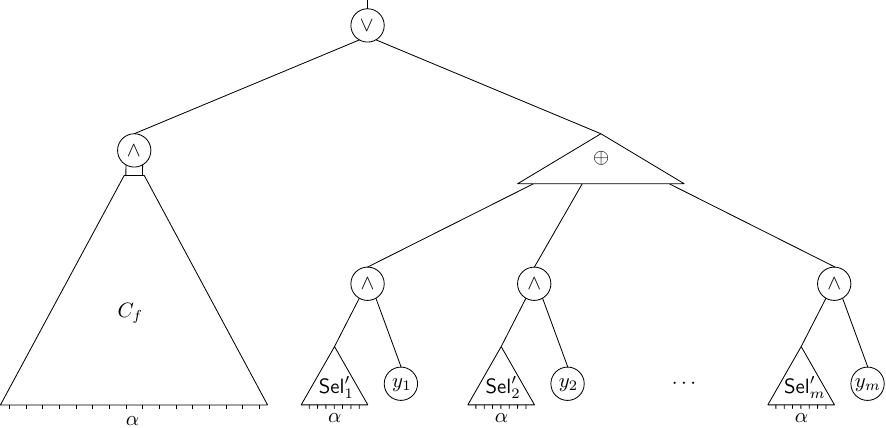}
    \caption{A schematic depiction of the formula after hitting it
      with the described restriction.}
    \label{fig:circuit-restricted}
  \end{figure}

  Note that this implements the intended semantics: for each
  input $\alpha \in \set{0,1}^n$ the selector circuits output $1$ on
  some variables $y_i$ which are then \emph{xor}'ed, and the
  restricted circuit outputs
  \begin{align}
    \bigoplus_{i \in N(\alpha)} y_i \eqcomma
  \end{align}
  unless $C_f(\alpha) \neq \bot$, in which case the output of the
  circuit is $f(\alpha)$ and all selector circuits output $0$. We
  require that $s$ is larger than the size of the described circuit
  which is of size $O\big(m \cdot \poly(n, k)\big) + s/2$.

  We have the intended semantics of the circuit and need to ensure the
  furthermore property: that the restricted formula is over few
  variables. First, since the selector circuits $\sel'_i$ are fixed,
  all evaluation variables for these subcircuits can be fixed to
  constants. The same holds for the circuit $C_f$. Similarly, since
  the $y_i$ gate always carries the value of the $y_i$ variable, all
  $2^n \cdot m$ wire variables corresponding to the $Y$ variables can
  be substituted by the corresponding $y_i$ variable and are thus
  restricted away.

  After these restrictions the only evaluation variables left are
  those for the evaluation of the $\oplus$ circuit.  For $\alpha$ such
  that $C_f(\alpha) \neq \bot$, the selector circuits are hard-wired
  to $0$ and in particular the inputs to the $\oplus$ circuit is
  hard-wired to $0$, meaning that these evalation variables can be
  restricted away.

  There remains then only the $O(t \cdot m)$ evaluation variables
  corresponding to the evaluation of the $\oplus$ circuit for inputs
  $\alpha$ such that $C_f(\alpha) = \bot$. Let us, without loss of
  generality, use an \emph{xor}-circuit which iteratively \emph{xor}s
  each variable.  Concretely, let it have subcircuits $\chi_i$ where
  $\chi_1 = \sel'_1 \wedge y_1$ and
  $\chi_{i} = \chi_{i-1} \oplus (\sel'_i \wedge y_i)$ for $i > 1$, and
  $\chi_m$ is the overall output of the $\oplus$ circuit.

  The only observation required is that if the circuit
  $\sel'_i(\alpha) = 0$, then $\chi_i$ gets a $0$ as input from index
  $i$, independent of the value of $y_i$. Hence the output wire
  variable of the circuit $\chi_i$ indexed by the input $\alpha$ can
  be substituted by the output of the circuit $\chi_{i-1}$. Hence for
  each $\alpha$ such that $C_f(\alpha) = \bot$, we can reduce the
  number of free wire variables indexed by $\alpha$ to $O(k)$, as each
  $\oplus$-constraint is over at most $k$ variables.  As $C_f$ outputs
  $\bot$ on at most $t$ inputs, we end up with a restriction leaving
  only a total of $O(t \cdot k + m)$ remaining variables in the
  restricted formula.

  This completes the description of the restriction $\rho$.  The only
  part that remains is to verify that $\rho$ is natural,
  $k$-determined, and $m$-independent.  That $\rho$ is natural is
  immediate -- it does not substitute any structural variable by an
  evaluation variable.  For $k$-determinedness, note that for a fixed
  input $\alpha$ at most $k$ selector circuits output $1$, and thus
  for every gate $u$ the value of $\outwire_\alpha(u)$ as a function
  of $Y$ can be computed by a function over those $k$ variables.
  Finally, each assignment to the remaining structure variables $Y$
  gives a valid circuit and thus $\rho$ is $m$-independent.
\end{proof}

We are ready to prove the degree lower bound, restated here for
convenience.

\MainDegree*

\begin{proof}
  Let $G = (U, V, E)$ be an explicit bipartite graph as in
  \cref{thm:expander}, with $U = \{0,1\}^n$,
  $k = O_{\gamma}\big((n \log r)^{1 + 1/\gamma}\big)$, and
  $|V| \le k^2 r^{1+\gamma}$ for parameters $\gamma > 0$ and
  $r \le 2^n$ to be fixed later. Apply \cref{lem:restriction} with $t=2^n$ along with
  $C_f = \bot$ to obtain, for $s \ge m \cdot \poly(n, k)$, a natural
  $m$-independent $k$-determined variable substitution $\rho$ for
  $\circuit_s(f)$ such that
  $\gout_{s, \alpha}(Y) = \oplus_{i \in N(\alpha)} y_i$. In words, the
  circuit of the restricted formula on input $\alpha$ computes an
  \emph{xor} of the neighborhood of the vertex $\alpha$ of $G$.

  Apply \cref{lem:circuit-to-poly} to $\rho$ to conclude that if there
  is an SoS refutation of $\restrict{\circuit_s(f)}{\rho}$ of degree
  $d$, then there is a degree $d \cdot k$ SoS refutation of the system
  of polynomial equations computing
  \begin{align}
    \nonumber
    \calP_G
    &= \Big\{
    \bigoplus_{i \in N(\alpha)} y_i = f(\alpha)
    :
    \alpha \in \set{0,1}^n
    \Big\}\,
    \cup
    \set{
    y_i^2 = y_i\mid i \in [m]
    }
    \eqperiod
  \end{align}
  As the graph $G$ is a strong expander, we can apply
  \cref{thm:sos-xor} to get an SoS degree lower bound of $\Omega(r)$
  for the XOR-CSP instance $\calP_G$ defined over $G$, which in turn
  gives us an $\Omega(r/k)$ degree lower bound for the
  $\restrict{\circuit_s(f)}{\rho}$ formula and hence also for the
  unrestricted formula.

  Let us fix the parameters. We want to choose $r$ as large as
  possible. However, the larger we choose $r$, the larger $m$ may
  become, since \cref{thm:expander} only guarantees that
  $m \le k^2 r^{1+\gamma}$. Let us analyze how large $r$ can be chosen
  in terms of $n$ and $s$.

  Note that $k = \poly_\gamma(n)$, where we use that $r \le 2^n$, and we write
  $\poly_\gamma(n)$ to denote some polynomial in $n$ whose degree and
  coefficients may depend on $\gamma$. Hence we may choose
  \begin{align}
    m = \frac{s}{\poly_{\gamma}(n)} \eqcomma
  \end{align}
  according to the requirement on $s$ in \cref{lem:restriction}.  From
  the guarantees of \cref{thm:expander} we know that
  $r \ge (m/k^2)^{1/(1+\gamma)}$. Substituting $m$ according to the
  previous equation we get that
  \begin{align}
    r
    \ge
    \left(
    \frac{s}{
    k^2 \poly_{\gamma}(n)
    }
    \right)^{\frac{1}{1 + \gamma}}
    =
    \frac{s^{1/(1 + \gamma)}}{\poly_{\gamma}(n)} \eqperiod
  \end{align}
  Hence if we choose $\gamma$ small enough so that
  $\frac{1}{1+\gamma} > 1-\eps/2$ and then require $s$ to be large
  enough such that the final $\poly_{\gamma}(n)$ is at most
  $s^{\eps/2}$, we obtain the claimed lower bound.
\end{proof}

In the following we prove the claimed size lower bound.

\MainSize*

\begin{proof}
  Apply \cref{lem:restriction} with the graphs from
  \cref{thm:expander} as in the proof of \cref{thm:main-degree}. We
  get a natural $m$-independent $k$-determined variable substitution
  $\rho$ and the formula $\restrict{\circuit_s(f)}{\rho}$ over
  $O(t \cdot k + m)$ variables.  To this formula we then apply
  \cref{lem:circuit-to-poly} to obtain a degree lower bound of
  $\Omega(r/k)$, akin to the proof of \cref{thm:main-degree}. By
  setting the parameters as in the aforementioned proof we get the
  same degree lower bound of $\Omega_\eps(s^{1-\eps/3})$ for the
  formula $\restrict{\circuit_s(f)}{\rho}$. As this formula is over
  few variables we can apply \cref{thm:tradeoff} to obtain an SoS size
  lower bound of
  $\exp \Big( \Omega_\eps\big((s^{1- \eps/3} - 3k)^2/(t \cdot k +
  m)\big) \Big)$ for the restricted formula. As variable substitutions
  may only decrease the size of a refutation, the same lower bound
  also holds for the unrestricted formula. We obtain the desired lower
  bound by choosing $s$ large enough such that
  $s^{\eps/3} \ge k = \poly_\eps(n)$ and by recalling that
  $t \ge s \ge m$. 
\end{proof}

%% file: monotone-lb.tex
\section{Lower Bounds for Monotone Circuits}
\label{sec:monotone-degree-LB}

Recall that $\calM_n(\ell)$ denotes all Boolean monotone $\ell$-slice
functions on $n$ bits: all Boolean functions
$f: \{0,1\}^n \rightarrow \{0,1\}$ that output $0$ on all inputs of
Hamming weight less than $\ell$ and $1$ on all inputs of Hamming
weight larger than $\ell$. There is no restriction on the output for
inputs of Hamming weight $\ell$ and we have
$|\calM_n(\ell)| = 2^{\binom{n}{\ell}}$.
Further, recall that $\calM_n(\ell, s, t) \subseteq \calM_n(\ell)$ is
the class of monotone Boolean $\ell$-slice functions
$f: \set{0,1}^n \rightarrow \set{0,1}$ for which there is a
(not necessarily monotone) Boolean circuit
$C_f^\mon: \set{0,1}^n \rightarrow \set{0,1, \bot}$ of size $s$ such
that
\begin{enumerate}
\item for all $\ell$-slice inputs $\alpha \in \binom{[n]}{\ell}$ it
  holds that if $C^\mon_f(\alpha) \neq \bot$, then
  $C^\mon_f(\alpha) = f(\alpha)$, and
\item $C^\mon_f(\alpha) = \bot$ on at most $t$ inputs
  $\alpha \in \binom{[n]}{\ell}$.
\end{enumerate}

It is very convenient to work with slice functions as we have a handle
on their monotone circuit complexity: by \cref{lem:slice} their
monotone circuit size is the same as their ordinary circuit size up to
a polynomial size increase. Hence we do not need to worry whether the
functions needed for the reduction have small monotone circuits, as
long as we are working on a slice only. 

The proof of the monotone lower bound is an adaption of the argument
used to prove
\cref{lem:restriction}. % We prove the monotone lower bounds by
% adapting the proof of \cref{lem:restriction}.
The idea is to work over
the $\ell$th slice and disregard all other inputs.  By
\cref{lem:slice} we can implement our selector circuits by small
monotone circuits. We then also need to take care of the negations in
the $\oplus$-circuit. We push the negations down until they either hit
a gate in $Y$ or a selector circuit. We create a set $\overline Y$
gates, which we can think of as the negation of the gates in $Y$ and
also create negated selector circuits (on the $\ell$th slice). By
doing so we can now get rid of the last negations by appropriately
connecting the appropriate circuits. The following corollary of
\cref{lem:slice} will be useful to us.

\begin{claim}
  \label{clm:to-slice}
  Let $C$ be a Boolean circuit on $n$ input bits of size $s$. Then,
  for $\ell \in [n]$, there is a monotone Boolean circuit $C^\mon$ of
  size $2s + \poly(n)$ computing the $\ell$-slice function that is
  equal to $C$ on the $\ell$-slice.
\end{claim}

\begin{proof}
  Let $\calT_{\ge \ell}$ be the threshold function that outputs $1$ if
  and only if the Hamming weight of an input $\alpha \in \set{0,1}^n$
  is at least $\ell$. Connect the output of $C$ by an \emph{and} gate
  to a circuit computing $\calT_{\ge \ell}$. The output of this
  circuit is then connected by an \emph{or} gate to the output of a
  circuit computing $\calT_{> \ell}$. Let us denote this new circuit
  by $C'$.

  The circuit $C'$ clearly outputs $1$ whenever the input is of
  Hamming weight larger than $\ell$. Furthermore, on the $\ell$-slice
  it is equal to $C$ because $\calT_{\ge \ell}$ outputs $1$ while
  $\calT_{> \ell}$ outputs $0$. Finally the output is $0$ if the
  Hamming weight is less than $\ell$ because the output of both
  threshold functions is $0$.

  Clearly the size of the circuits computing the threshold functions
  is $\poly(n)$. We apply \cref{lem:slice} to conclude that there is a
  monotone circuit $C^\mon$ computing the same function as $C'$ of
  size $2s + \poly(n)$.
\end{proof}

Before stating the following lemma we need to adapt some terminology
to the monotone setting. Observe that $\circuit_s^\mon(f)$ is a
restriction of $\circuit_s(f)$. Let $\tau$ be such that
$\restrict{\circuit_s(f)}{\tau} = \circuit^\mon_s(f)$. This allows us
to naturally extend $k$-determined restrictions to the monotone
setting: a restriction $\rho$ is a $k$-determined restriction for
$\circuit^\mon_s(f)$ if the restriction $\rho \tau$ is a
$k$-determined restriction for $\circuit_s(f)$. Similarly we can
extend $m$-independence to the monotone setting. This will later allow
us to use \cref{lem:circuit-to-poly} even though we are working with
the monotone formula.

\begin{lemma}
  \label{lem:restriction-monotone}
  For all $k, \ell, m, n, t \in \N$ satisfying $m \le 2^n$, and any
  explicit bipartite graph $G = (U, V, E)$ such that $|U| = 2^n$,
  $|V| = m$ and all $u \in U$ are of degree $\deg(u) \le k$, the
  following holds.  There is a constant $C > 0$, depending on the
  explicitness of $G$, such that for all
  $s \ge C \cdot m \cdot n^C \cdot k^C$ and any
  $f \in \calM_n(\ell, s/10, t)$ there is a natural $m$-independent
  $k$-determined variable substitution $\rho$ for the formula
  $\circuit^\mon_s(f)$ such that
  \begin{align*}
    \gout_{s, \alpha}(Y) =
    \begin{cases}
      1, &\text{if~} |\alpha| > \ell,\\
      0, &\text{if~} |\alpha| < \ell,\\
      f(\alpha), &\text{if~} |\alpha| = \ell \text{~and~} C_f^\mon(\alpha) \neq \bot,\\
      \oplus_{i \in N(\alpha)} y_i, &\text{otherwise},\\
    \end{cases}
  \end{align*}
  for $\gout_{s, \alpha}$ and $Y$ as in \cref{def:determined}.

  Furthermore, the formula $\restrict{\circuit^\mon_s(f)}{\rho}$ is
  over $O(t \cdot k + m)$ variables.
\end{lemma}

\begin{proof}
  This proof is an adaptation of the argument of the proof
  \cref{lem:restriction}. Let us describe the natural $m$-independent
  $k$-determined restriction $\rho$ for the formula
  $\circuit_s^\mon(f)$.
  
  As in the proof of \cref{lem:restriction} we have gates that act as
  Boolean variables. But instead of having a single set $Y$ of
  variables we now have two sets $Y$ and $\overline Y$, each of size
  $m$. We think of the variables in $\overline Y$ as the negations of
  the variables in $Y$ and ensure this by applying the appropriate
  variable substitution for all $\alpha \in \set{0,1}^n$ and $i \in [m]$.

  According to \cref{clm:to-slice} we may assume that the circuit
  $C_f^\mon$ computes a monotone $\ell$-slice function in both outputs
  $C_{f,1}^\mon, C_{f,2}^\mon$ for a mild increase in size;
  $|C_f^\mon| \le s/5 + \poly(n) \le s/4$ for $s$ large enough. Recall
  that the first output of $C_f^\mon$ indicates whether the second
  output bit is equal to $f$ on the $\ell$-slice. Let
  $\overline C^\mon_{f,1}$ be the negation of $C_{f,1}^\mon$ on the
  $\ell$-slice. In other words,
  $\overline C^\mon_{f,1}(\alpha) = \lnot C_{f,1}^\mon(\alpha)$ if
  $\alpha$ has Hamming weight $\ell$, and
  $\overline C^\mon_{f,1}(\alpha) = C_{f,1}^\mon(\alpha)$ otherwise.

  The monotone circuit $C_f^\mon$ is of size at most $s/4$ and hence
  according to \cref{lem:slice} there is a monotone circuit of size
  $s/2 + \poly(n) \le 5s/8$ computing
  $\overline C^\mon_{f,1}(\alpha)$.

  We restrict the formula such that a part of the circuit is
  equivalent to $C_f^\mon$ and another part is equal to
  $\overline C^\mon_{f,1}$. Note that the size of these two circuits
  is at most $7s/8$ by above discussion.

  Recall that because $G(\set{0,1}^n, Y, E)$ is explicit, there are
  circuits $\sel_1,$ $\sel_2, \ldots, \sel_m$, each of size $\poly(n,k)$, where
  each $\sel_i$ computes, given an input $\alpha \in \set{0,1}^n$,
  whether the vertex $y_i \in Y$ is a neighbor of the vertex $\alpha$.
  Let $\overline \sel_i = \lnot \sel_i$ and denote by $\sel_i^\mon$
  (respectively $\overline \sel_i^\mon$) the circuit obtained by
  applying \cref{clm:to-slice} to $\sel_i$ (to $\overline \sel_i$
  respectively). By the guarantees of \cref{clm:to-slice} all these
  $2m$ circuits are of size $\poly(n,k)$.

  We restrict the formula such that a part of the circuit computes the
  functions
  \begin{align}
    \sel_1^\mon, \ldots, \sel^\mon_m,
    \overline\sel^\mon_1, \ldots, \overline\sel^\mon_m \eqperiod
  \end{align}
  From these $\ell$-slice selector circuits we can then define
  selector circuits that take $C_f^\mon$ into account. Namely, we
  connect $\sel_i^\mon$ by an \emph{and} gate to the output of
  $\overline C_{f,1}^\mon$ to obtain the circuit $\sel_{i}^{\prime\mon}$ and
  similarly connect $\overline \sel^\mon_i$ by an \emph{or} gate to
  $C_{f,1}^\mon$ to obtain the circuit $\overline \sel_i^{\prime\mon}$. 

  Finally, we also put each variable $y_i$ and $\bar y_i$ onto the
  slice by the same construction used in the proof of
  \cref{clm:to-slice}: connect the variable $y_i$ (respectively
  $\bar y_i$) by an \emph{and} to the threshold circuit
  $\calT_{\ge \ell}$ and connect this circuit in turn by an \emph{or}
  gate to a $\calT_{> \ell}$ threshold circuit to obtain $y_i^\mon$
  (respectively $\bar y_i^\mon$).  It is well-known \cite{Valiant84,
    BW06, Goldreich2020} that threshold circuits have montone circuits
  of size $\poly(n)$ and we can thus restrict the formula such that a
  part of the circuit computes $y_i^\mon$ and $\bar y_i^\mon$.
  
  Finally we connect $y^\mon_i$ by an \emph{and} gate to the selector
  circuit $\sel_i^{\prime\mon}$. Note that this circuit is equal to an
  $\ell$-slice function. As we will see later this ensures that the
  whole circuit outputs an $\ell$-slice function.
  We connect the circuits $\bar y^\mon_i$ similarly: connect
  $\bar y^\mon_i$ by an \emph{or} gate to the negated selector circuit
  $\overline \sel_i^{\prime\mon}$. Again, the output of this circuit
  is equal to an $\ell$-slice function.

  Equally inportant is that these circuits behave well on the
  $\ell$-slice. Indeed it can be checked that the positive circuit, on
  input $\alpha \in \set{0,1}^n$, outputs
  $\sel_i^{\prime\mon}(\alpha) \land y_i$ while the negative circuit
  outputs $\overline \sel_i^{\prime\mon}(\alpha) \lor \bar y_i$. On
  the $\ell$-slice these functions are the negation of eachother,
  which we are going to use in the following.
  
  We need to construct a monotone circuit for the \emph{xor} of
  $\sel_i^{\prime\mon}(\alpha) \wedge y_i$ for $i$ from $1$ to $m$, on
  $\ell$-slice inputs $\alpha$. We take a standard $O(m)$-size
  $\oplus$-circuit and monotonize it by pushing all negations in it
  down using De Morgan's law until they reach one of the
  $\oplus$-circuit's inputs $\sel_i^{\prime\mon} \wedge y_i$.
  Whenever the negation of $\sel_i^{\prime\mon}(\alpha) \wedge y_i$ is
  needed, we do one last step of De Morgan and replace it by
  $\overline \sel_i^{\prime\mon}(\alpha) \vee \bar{y}_i$.

  To ensure that the circuit outputs $f(\alpha)$ whenever
  $C^\mon_f(\alpha) \neq \bot$, we connect the two outputs of
  $C_f^\mon$ by an \emph{and} gate and connect this gate by an
  \emph{or} gate to the output of the \emph{xor} circuit. This
  completes the description of the restriction on the structure
  variables. A depiction of the resulting circuit can be found in
  \cref{fig:monotone}.
  \begin{figure}
    \centering
    \includegraphics[width=\linewidth,height=\textheight,keepaspectratio]{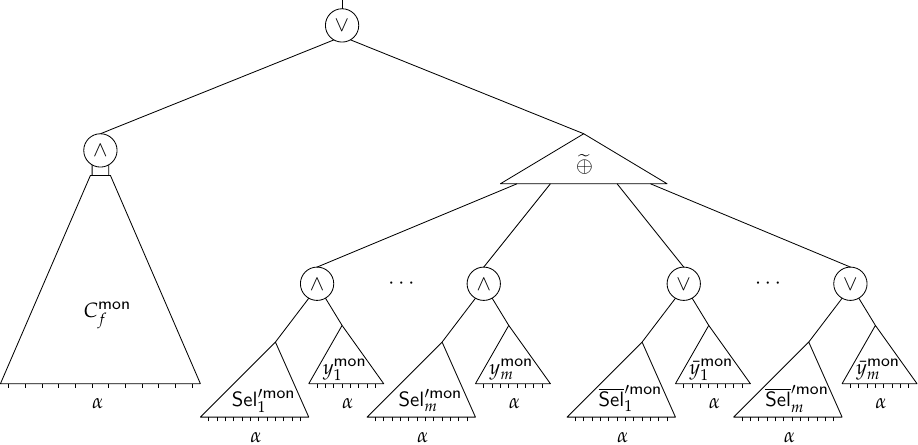}
    \caption{A depiction of the monotone circuit, where
      $\widetilde\oplus$ is the $\oplus$ circuit with the negations
      pushed down.}
    \label{fig:monotone}
  \end{figure}
  We ensure that $s$ is large enough so that above circuit can be
  described by the formula.

  Note that the constructed circuit always outputs a monotone
  $\ell$-slice function: as the monotonized $\oplus$-circuit is
  non-constant, we see that if all inputs to the circuit are $0$, it
  outputs $0$ and if all inputs are $1$, it outputs $1$. This, in
  particular, implies that the circuit outputs $0$ (respectively $1$)
  if the input is below (respectively, above) the $\ell$-slice and
  hence the entire circuit computes a monotone $\ell$-slice function.

  It can be easily checked that the described restriction is
  $m$-independent and $k$-determined. In order to prove the
  furthermore part, we need to reduce the number of evaluation
  variables. This can be achieved analogous to the proof of
  \cref{lem:restriction} and we thus omit it here.
\end{proof}

Let us prove our degree lower bound for monotone circuits, restated
here for convenience.

\MonotoneDegree*

\begin{proof}[Proof of \cref{thm:main-monotone-degree}]
  As in the proof of \cref{thm:main-degree}, we use the graphs from
  \cref{thm:expander}, with $U = \{0,1\}^n$,
  $k = O_{\gamma}\big((n \log r)^{1 + 1/\gamma}\big)$, and
  $|V| \le k^2 r^{1+\gamma}$ for parameters $\gamma > 0$ and
  $r \le 2^n$. We apply \cref{lem:restriction-monotone} with above
  graph and $t=2^n$ along with $C_f^\mon = \bot$ to obtain, for
  $s \ge m \cdot \poly(n, k)$, an appropriate natural $m$-independent
  $k$-determined variable substitution $\rho$ for
  $\circuit^\mon_s(f)$. In particular $\rho$ satisfies
  \begin{align*}
    \gout_{s, \alpha}(Y) =
    \begin{cases}
      1, &\text{if~} |\alpha| > \ell,\\
      0, &\text{if~} |\alpha| < \ell,\\
      \oplus_{i \in N(\alpha)} y_i, &\text{otherwise},\\
    \end{cases}
  \end{align*}
  for $\gout_{s, \alpha}$ and $Y$ as in definition
  \cref{def:determined}.

  Recall that there is a restriction $\tau$ such that
  $\circuit_s^\mon(f) = \restrict{\circuit_s(f)}{\tau}$ and we can
  thus apply \cref{lem:circuit-to-poly} with $\tau\rho$ to conclude
  that if there is an SoS refutation of
  $\restrict{\circuit^\mon_s(f)}{\rho}$ in degree $d$, then there is a
  degree $d \cdot k$ SoS refutation of the system of polynomial
  equations computing
  \begin{align}
    \set{
    \bigoplus_{i \in N(\alpha)} y_i = f(\alpha)
    \mid
    \alpha \in \binom{[n]}{\ell}
    }
    \eqperiod
  \end{align}
  As the graph $G$ is a strong expander, we can apply
  \cref{thm:sos-xor} to get an SoS degree lower bound of $\Omega(r)$
  for above system of equations. By above connection this gives an
  $\Omega(r/k)$ degree lower bound for the
  $\restrict{\circuit^\mon_s(f)}{\rho}$ formula and hence also for the
  unrestricted formula.

  Regarding the parameters, as in the proof of \cref{thm:main-degree}
  we choose $m = s/\poly_\gamma(n)$. Repeating the
  calculations from the aforementioned proof we obtain that
  $r \ge s^{1/(1+\gamma)}/\poly_\gamma(n)$. Thus by choosing
  $\gamma$ small enough such that $\frac{1}{1 + \gamma} > 1 - \eps/2$
  and $s$ large enough such that the final
  $\poly_\gamma(n) \le s^{\eps/2}$ we obtain the claimed degree lower
  bound of $\Omega_\eps(s^{1-\eps})$.
\end{proof}

As in the non-monotone case, we can also obtain size lower bounds for
functions that almost have a circuit of size $s$.

\MonotoneSize*

\begin{proof}
  Analogous to the proof of \cref{thm:main-size}.
\end{proof}

%% file: upper-bound.tex
\section{Degree Upper Bound}
\label{sec:upper-bound}

In this section we give a simple upper bound on the SoS refutation
degree for $\circuit_s(f)$.  Specifically, for functions $f$ that have
no circuit of size $s$, we show that there is an SoS refutation of
$\circuit_s(f)$ of degree $O(s)$, essentially matching our
$\Omega(s^{1-\epsilon})$ lower bound.

At a high level, the logic behind the refutation is as follows: first,
we show that SoS in degree $O(s)$ can derive that the $\Theta(s^2)$
structure variables of $\circuit_s(f)$ must uniquely describe a
circuit of size $s$.  Then, we also show that for any fixed circuit,
SoS can in degree $O(s)$ derive that the circuit does not compute $f$
by ``evaluating'' the circuit on some (non-deterministically chosen)
input where the output differs from $f$.

To make this precise, we first define a set of monomials, which correspond
to circuits of size $s$. A multilinear monomial $m$ is a \emph{circuit
  monomial} if for every gate $v \in [s]$ it holds that
\begin{enumerate}
\item exactly one of the variables $\isneg(v), \isor(v)$ or $\isand(v)$ occurs
  in $m$,
\item for $a \in \set{1,2}$ exactly one of the variables
  $\isfromconst(v, a)$, $\isfromvar(v, a)$ or $\isfromgate(v, a)$ occurs
  in $m$,
\item for $a \in \set{1,2}$ exactly one of the variables
  $\constval(v, a)$ or $\overline{\constval(v, a)}$ occurs in $m$,
\item for $a \in \set{1,2}$ exactly one of the variables
  $\set{\isvar(v,a,i) \mid i \in [n]}$
  occurs in $m$, and
\item for $a \in \set{1,2}$ and $v > 1$, exactly one of the variables
  $\set{\isgate(v,a,u) \mid u < v}$
  occurs in $m$, and
  \item no other variables occur in $m$ than the ones described
    above.
\end{enumerate}
We denote by $\calM_s$ the set of circuit monomials. We first show
that SoS can derive in degree $O(s)$ the polynomial
$\sum_{m \in \calM_s} m - 1$; this corresponds to SoS proving that
the structure variables uniquely describe a circuit and does not use
anything about $f$.  Then in a second step we show that for every
$m \in \calM_s$, SoS can derive $-m$ in degree $O(s)$; this
corresponds to SoS proving that the circuit described by $m$ does not
compute $f$ correctly.  Summing these two parts up yields an SoS
derivation of $-1$, i.e., a refutation of the $\circuit_s(f)$ formula.

\pa{Here we still talk about SoS  derivations so we can't remove it from prelims}

\paragraph{Deriving~$\bm{\sum_{m \in \calM_s} m - 1}$.} We proceed
by induction on $s$. Note that $\calM_0 = \set{1}$ and hence the base
case is trivial. Suppose we have an SoS derivation of
$\sum_{m \in \calM_s} m - 1$. For every monomial $m \in \calM_s$ we
add the polynomial
\begin{align}
  m \cdot \big(\isneg(v) +  \isor(v) + \isand(v) - 1\big)
\end{align}
to the derivation (note that the second term is \cref{eq:ax-fn}). This
gives us an SoS derivation of $\sum_{m \in \calM'_s} m - 1$, where
$$\calM'_s = \bigcup_{m \in \calM_s} \set{m \cdot \isneg(v), m \cdot \isor(v), m \cdot \isand(v)}\eqperiod$$ We can continue in the same manner with
\cref{eq:ax-from} and \cref{eq:ax-in-gt} to finally obtain an SoS
derivation of $\sum_{m \in \calM_{s+1}} m - 1$. Clearly this
derivation requires degree at most $O(s)$, as for each gate there are
at most $7$ variables in every monomial from $\calM_s$.

\paragraph{Deriving $-m$ for $m \in \calM_s$.}
Let $C$ be the circuit that corresponds to the monomial $m$ and let
$\alpha \in \set{0,1}^n$ be such that $f(\alpha) \neq C(\alpha)$ (by
the assumption that $f$ does not have a circuit of size $s$, such an
$\alpha$ exists).  Suppose $C(\alpha) = b$ but $f(\alpha) = 1-b$.

We construct a degree $O(s)$ SoS proof of the fact that
$C(\alpha) = b$. That is, we are going to show that the polynomial
$p_s = m \cdot \left( \outwire_\alpha(s) - b\right)$ can be written as
\begin{align}
  \label{eq:-m-derivation}
  p_s = \sum_{i=1}^t r_i \cdot q_i \eqcomma
\end{align}
for some parameter $t$, axioms $q_1, \ldots, q_t$ and some polynomials
$r_1, \ldots, r_t$, such that $\deg(r_i \cdot q_i) = O(s)$ for all
$i$.  Note that given this, we can then easily derive $-m$ by
subtracting the polynomial
$m \cdot \left( \outwire_\alpha(s) - f(\alpha)\right) = m \cdot \left(
  \outwire_\alpha(s) - 1 + b\right)$ (this is $m$ multiplied by
\cref{eq:ax-correct}), yielding a derivation of $m \cdot (1 - 2b)$
which is either $m$ or $-m$ depending on $b$; in the former case it
can be multiplied by $-1$ to yield $-m$.  Note that here it is
important that the derivation \eqref{eq:-m-derivation} is a
Nullstellensatz derivation, not using any Sum-of-Squares part, since
otherwise it would not be possible to multiply it by $-1$.

Let us thus see how to derive $m \cdot \left( \outwire_\alpha(s) - b\right)$.
We do this by structural induction over the circuit: for every gate
$v$ we are going to construct an SoS proof of the fact that the
circuit rooted at $v$ outputs the bit $b_v$ on input $\alpha$. In
other words, an SoS derivation of the polynomial
$m \cdot \left(\outwire_\alpha(v) - b_v\right)$.

Let us explain how to construct an SoS proof $p_v$. Consider a gate
$v$ in the circuit. Depending on the function computed at $v$ and how
the wires of $v$ are connected we construct $p_v$ slightly differently.
As a first step let us construct SoS proofs $q_1$ and $q_2$ of the
fact that on input $\alpha$ the bit $c_a$, $a \in \set{1,2}$, is
carried on wire $a$ to the gate $v$. That is, the polynomial $q_a$
should simplify to $m \cdot \big(\inwire_\alpha(v, a) - c_a\big)$. In
the following we explain how to precisely define $q_a$ depending on
what the wire is connected to. Note that not a lot is going on -- we
are mostly just multilinearizing using the Boolean axioms.

If $m$ is of the form $m = m' \cdot \isfromconst(v,a) \cdot \constval(v, a)$ for some monomial $m'$, i.e., wire $a$
is connected to the constant $c_a = 1$ in the circuit described by $m$, then we
can derive $q_a = m \cdot \big(\inwire_{\alpha}(v,a) - 1\big)$
by the identity
\begin{align}
  \nonumber
  q_a &= m' \cdot \constval(v,a) \cdot \isfromconst(v,a) \cdot \big(\inwire_{\alpha}(v,a) -
        \constval(v, a)\big) +\\
      &\qquad m' \cdot \isfromconst(v,a) \cdot
        \big( \constval(v,a)^2 - \constval(v,a)\big) \eqcomma
\end{align}
a linear combination of \cref{eq:inwire-const} and the Boolean axiom on $\constval(v,a)$.  Similarly if $m = m' \cdot \isfromconst(v,a) \cdot \overline{\constval(v,a)}$ (i.e., $c_a = 0$) we can derive $q_a = m \cdot \inwire_{\alpha}(v,a)$ by
\begin{align}
  \nonumber
  q_a &= m' \cdot \overline{\constval(v,a)} \cdot \isfromconst(v,a) \cdot \big(\inwire_{\alpha}(v,a) -
        \constval(v, a)\big) +\\\nonumber
      &\qquad m' \cdot \isfromconst(v,a) \cdot \constval(v,a) \cdot
         \\\nonumber
      &\qquad\qquad\qquad \big( 1 - \constval(v,a) - \overline{\constval(v,a)}\big) +\\
      &\qquad m' \cdot \isfromconst(v,a) \cdot (\constval(v,a)^2 - \constval(v,a))
        \eqcomma
\end{align}
where we additionally use the negation axiom on $\constval(v,a)$.

Next, if $a$ is connected to an input $i$ and $m$ is of the form
$m = m' \cdot \isfromvar(v, a) \cdot \isvar(v, a, i)$ (so that $c_a = \alpha_i$), then
\begin{align}
  q_a = m \cdot (\inwire_{\alpha}(v,a) - \alpha_i) &= m' \cdot \isfromvar(v, a) \cdot \isvar(v, a, i) \cdot
          \big( \inwire_\alpha(v, a) - \alpha_i\big) \eqcomma
\end{align}
which is a multiple of \cref{eq:inwire-var}.  Lastly, if
$a$ is connected to a gate $u$ and $m$ is of the form
$m = m' \cdot \isfromgate(v, a) \cdot \isgate(v, a, u)$ (i.e., $c_a = b_u$), then
\begin{align}
  \nonumber
  q_a = m \cdot(\inwire_\alpha(v,a) - b_u) &= m' \cdot
          \isfromgate(v, a) \cdot \isgate(v, a, u) \cdot
          \big(\inwire_\alpha(v, a) - \outwire_\alpha(u)\big)\\
                                           &\qquad+ m \cdot (\outwire_\alpha(u) - b_u)
                                             \eqcomma
\end{align}
where the first term is a multiple of \cref{eq:inwire-gate}, and the
second term is the polynomial $p_u$ which, by induction, we assume has
already been derived in degree $O(s)$.

Given the two SoS proofs $q_1$ and $q_2$ we are ready to construct the
SoS proof $p_v = m \cdot (\outwire_\alpha(v) - b_v)$. As mentioned earlier
we do a case distinction over the funtion computed at $v$.
\begin{enumerate}
\item \emph{$v$ is a \emph{not} gate ($m = m' \cdot \isneg(v)$ and $b_v = 1-c_1$).} We have the derivation
  \begin{align}
    \nonumber
    p_v &=  m' \cdot \isneg(v) \cdot
             \big(\outwire_{\alpha}(v) - \overline{\inwire_{\alpha}(v, 1)}\big) \\\nonumber
           &\qquad + m \cdot
             \big(\overline{\inwire_{\alpha}(v, 1)} - 1 + \inwire_{\alpha}(v, 1)\big)\\
           &\qquad - q_1\eqcomma
  \end{align}
  where the first line uses \cref{eq:ax-neg} and the
  second line uses the negation axiom for $\inwire_{\alpha}(v, 1)$.
\item \emph{$v$ is an \emph{or} gate ($m = m' \cdot \isor(v)$, and $b_v = 1-(1-c_1)(1-c_2)$.} We have the derivation
  \begin{align}
    \nonumber
    p_v =  & m' \cdot \isor(v) \cdot 
             \Big(
             \outwire_{\alpha}(v) -
             \big(
             1 -
             \overline{\inwire_{\alpha}(v, 1)} \cdot
             \overline{\inwire_{\alpha}(v, 2)}
             \big)
             \Big)\\\nonumber
           & -
             m \cdot
             \overline{\inwire_\alpha(v, 1)} \cdot
             \big(
             \overline{\inwire_{\alpha}(v, 2)}
             - 1 + \inwire_{\alpha}(v, 2)
             \big)\\\nonumber
           & +
             m \cdot
             \big(
             \inwire_\alpha(v, 2) - 1
             \big)
             \cdot
             \big(
             \overline{\inwire_{\alpha}(v, 1)}
             - 1 + \inwire_{\alpha}(v, 1)
             \big)\\\nonumber
           & +
             \big(
             1 - \inwire_\alpha(v, 1)
             \big)
             \cdot
             q_2 \\
           & +
             \big(
             1 - c_2
             \big)
             \cdot
             q_1\eqcomma
  \end{align}
  where the first line uses \cref{eq:ax-or}, and the
  following two lines uses negation axioms.
\item \emph{$v$ is an \emph{and} gate ($m = m' \cdot \isand(v)$, and $b_v = c_1 \cdot c_2$.} We have
  \begin{align}
    \nonumber
    p_v =  & m' \cdot \isand(v) \cdot
             \big(
             \outwire_{\alpha}(v) -
             \inwire_{\alpha}(v, 1) \cdot
             \inwire_{\alpha}(v, 2)
             \big)\\\nonumber
           & +
             \inwire_\alpha(v, 1)
             \cdot
             q_2 \\
           & +
             c_2
             \cdot
             q_1\eqcomma
  \end{align}
  where the first line uses \cref{eq:ax-and}.
\end{enumerate}
This completes the description of the SoS derivation of
$\outwire_\alpha(s) = b$.
Observe that the final proof $p_s$ is of degree $O(s)$: in every
inductive step we increase the degree of the proof by at most a
constant.

%% file: conclusion.tex
\section{Concluding Remarks}
\label{sec:conclusion}

We have shown degree and size lower bounds in the Sum-of-Squares proof
system for the minimum circuit size problem.  There are a number of
interesting questions left open for further study.  Let us name a few.

\paragraph{Better Size Lower Bounds}

Whereas our degree lower bounds apply for all Boolean functions $f$, the
corresponding size lower bounds only apply to an albeit rich but still
restricted class of functions.  

% \paragraph{Other Proof Systems}

% It would be very interesting to obtain results for the Cutting Planes
% proof system.  This seems to boil down to proving strong new monotone
% circuit lower bounds, via the connection of \cite{FPPR17RandomCNFs, HP17RandomCNFs}.

\paragraph{Monotone Circuit Lower Bounds}

For monotone circuits, we were only able to obtain lower bounds for
slice functions (essentially because they behave in many ways like
non-monotone functions).  An intriguing question is whether this
limitation can be overcome, or whether it is inherent and there exist
some monotone circuit lower bounds that SoS \emph{is} able to prove.

% \paragraph{Bounded Depth aka NC$_0$}
% If the circuit is bounded by depth $d$, then the polynomial computing
% the function has degree at most $\exp(d)$. This holds because each
% gate has bounded fan-in.

% \pa{I don't know what this paragraph wants to say.}

%% file: cnf-encoding.tex
\section[On Encodings of the $\circuit_s(f)$ Tautology]{On Encodings
  of the $\bm{\mathrm{Circuit}_s(f)}$ Tautology}
\label{sec:encoding}

Let us introduce a possible constant width CNF encoding of
$\circuit_s(f)$ as proposed by Razborov \cite{Razborov98}.

The formula is defined over the same variables as introduced in
\cref{sec:circuit-encoding}, but in order to keep the fan-in bounded,
we further introduce the extension variables $\isvarless(v, a, i)$ and
$\isgateless(v, a, u)$ that indicate whether the wire $a$ of $v$ is
connected to a variable in $x_1, \ldots, x_i$, a gate $1, \ldots, u$
respectively.

Let us group the axioms in the same manner as we did in
\cref{sec:circuit-encoding}. First we have the structure axioms. The
first axioms encode that each wire is connected to a single kind
\begin{align}
  &\big(\isfromconst(v, a) \vee \isfromvar(v, a) \vee \isfromgate(v,
    a)\big) \wedge \nonumber\\
  &\neg\big(\isfromconst(v, a) \wedge \isfromvar(v, a)\big) \wedge\nonumber\\
  &\neg\big(\isfromvar(v, a) \wedge \isfromgate(v, a)\big) \wedge\nonumber\\
  &\neg\big(\isfromconst(v, a) \wedge \isfromgate(v, a)\big) \eqperiod
\end{align}
The next set of axioms similarly ensures that each gate computes
precisely one function
\begin{align}
  &\big(\isneg(v) \vee \isor(v) \vee \isand(v)\big) \land\nonumber\\
  &\neg\big(\isneg(v) \wedge \isor(v)\big) \wedge \neg\big(\isor(v) \wedge
    \isand(v)\big) \wedge \neg\big(\isneg(v) \wedge \isand(v)\big) \eqperiod
\end{align}
Last, we need to make sure that each wire is connected to a single
input or a gate.
\begin{align}
  \label[axiom]{ax:input-cnf}
  &\isvarless(v, a, n) \wedge \bigwedge_{i \neq j}\neg \big(\isvar(v,a,i) \wedge
    \isvar(v,a,j)\big) \wedge\nonumber\\
    &\bigwedge_{i \in [n]}
      \Big(\isvarless(v, a, i) \equiv
    \big(\isvarless(v, a, i-1) \vee \isvar(v, a, i)\big)
    \Big) \eqcomma \nonumber\\
  &\qquad\text{where } \isvarless(v, a, 0) \defeq 0 \eqcomma
\end{align}
and similarly for $v \in [s] \setminus \{1\}$ we have that 
\begin{align}
  \label[axiom]{ax:gate-cnf}
  &\isgateless(v, a, v-1) \wedge
    \bigwedge_{u < u' < v} \neg\big(\isgate(v,a,u) \wedge \isgate(v,a,u')\big)
    \wedge\nonumber\\
  &\bigwedge_{u \in [v-1]}
    \Big(
    \isgateless(v, a, u) \equiv
    \big(\isgateless(v, a, u-1) \vee \isgate(v, a, u)\big)
    \Big)\eqcomma \nonumber\\
  &\qquad\text{where }\isgateless(v, a, 0) \defeq 0 \eqperiod
\end{align}

Let us take a look at the evaluation axioms. Again, we have axioms
that ensure that the wires carry the values intended by the structure
variables. If a wire is connected to a constant, then the evaluation
variable associated with that wire should be equal to the constant
\begin{align}
  &\isfromconst(v, a)
    \rightarrow
    \big(\inwire_{
    \alpha}(v, a) \equiv \constval(v, a)\big) \eqcomma
\end{align}
and similarly if a wire is connected to an input or a gate
\begin{align}
  &\isfromvar(v, a) \wedge \isvar(v, a, i)
    \rightarrow
    \inwire_{\alpha}(v, a) \equiv \alpha_i \eqcomma\\
  &\isfromgate(v, a) \wedge \isgate(v, a, u)
    \rightarrow
    \inwire_{\alpha}(v, a) \equiv \outwire_{\alpha}(u) \eqperiod
\end{align}
Last we need to make sure that the gates propagate the value they are
supposed to compute.
\begin{align}
  &\isneg(v) \rightarrow
    \big(\outwire_{\alpha}(v) \equiv \neg \inwire_{\alpha}(v, 1)\big)\\
  &\isor(v) \rightarrow
    \big(\outwire_{\alpha}(v) \equiv \inwire_{\alpha}(v, 1) \vee \inwire_{\alpha}(v, 2)\big)\\
  &\isand(v) \rightarrow
    \big(\outwire_{\alpha}(v) \equiv \inwire_{\alpha}(v, 1) \wedge
    \inwire_{\alpha}(v, 2)\big)
    \eqperiod
\end{align}

The final axioms ensure that the correct function is computed
\begin{align}
  \outwire_{\alpha}(s) \equiv f(\alpha) \eqperiod
\end{align}

This formula can be rewritten in the usual manner into a $4$-CNF. Let
us denote this formula by $\circuit^{\mathrm{CNF}}_s(f)$.

\pa{Here the $4$ is because all axioms use at most 4 variables, right?  Can say this.}

\kr{It is a bit annoying: you have to consider it as a CNF and then
  each clause is over at most $4$ vars. The axioms themselves (for
  example \cref{ax:gate-cnf}) may depend on more variables. }

\begin{proposition}
  \label{prop:cnf-poly}
  If there is an SoS refutation of degree $d$ of the CNF formula
  $\circuit^{CNF}_s(f)$, then there is an SoS refutation of degree
  $O(d)$ of the system of polynomials $\circuit_s(f)$ as introduced in
  \cref{sec:circuit-encoding}.
\end{proposition}

The rest of this section is devoted to the proof of
\cref{prop:cnf-poly}.

Observe that for each axiom $p$ from the polynomial encoding
$\circuit_s(f)$, there is a CNF
$F_p \subseteq \circuit_s^{\mathrm{CNF}}(f)$ over the same variables
as $p$ (ignoring the added extension variables $\isvarless(v, a, i)$
and $\isgateless(v, a, u)$) such that $p(\alpha) = 0$ is satisfied by
a Boolean assignment $\alpha$ if and only if $F_p$ is satisfied by
$\alpha$ (where we extend the assignment to the extension variables in
the natural manner).

\pa{``extension variables'' not defined.  When introducing them, say
  that they are extension variables, and remind here what the
  extension variables are.  But I also don't understand what the
  purpose of this paragraph is, why is it needed?}

\kr{added the definition of extension variables and remind the reader
  by adding the names. The purpose of this paragraph is to introduce
  the notation $F_p$.}

Recall that SoS operates on polynomials and we thus need to translate
the CNF into a system of polynomials. We translate a clause
$\vee_{i \in [w]} z_i$ into the polynomial
$\prod_{i \in [w]}(1 - z_i) = 0$.

Observe that almost all axioms $p$ of $\circuit_s(f)$ depend only on a
constant number of variables. From such $p$, using the appropriate
Boolean axioms and negation axioms, we can in constant degree derive
$F_p$.

Let us define a polynomial substitution $\rho$ that gets rid of the
extension variables in the natural manner: the substitution $\rho$
first replace each occurrence of a ``bar'' extension variable
$\overline{\isvarless(v, a, i)}$ or $\overline{\isgateless(v, a, u)}$
by the polynomial $1-\isvarless(v, a, i)$ and the polynomial
$1-\isgateless(v, a, u)$ respectively. Then, $\rho$ replaces each
occurrence of the variable $\isvarless(v, a, i)$ by
$\sum_{j \le i} \isvar(v, a, j)$ and similarly $\isgateless(v, a, u)$
by $\sum_{w \le u} \isgate(v, a, w)$.

Suppose we have a degree $d$ refutation $\pi$ of
$\circuit_s^{\mathrm{CNF}}(f)$. Let us consider $\restrict{\pi}{\rho}$
and $\restrict{\circuit^{\mathrm{CNF}}(f)}{\rho}$. Note that
$\restrict{\pi}{\rho}$ is a degree $d$ SoS refutation of
$\restrict{\circuit^{\mathrm{CNF}}_s(f)}{\rho}$.

We claim that in constant degree the axioms of
$\restrict{\circuit^{\mathrm{CNF}}_s(f)}{\rho}$ can be derived from the
polynomial encoding $\circuit_s(f)$. As previously noted, this holds
for all axioms but the ones that are over a non-constant number of
variables. In other words it just remains to show that we can derive
the substituted \cref{ax:input-cnf,ax:gate-cnf} from
\cref{eq:ax-in-gt,eq:ax-in-gt2,eq:prod-zero-in,eq:prod-zero-gate}.

Let us consider \cref{ax:input-cnf}. With the extension variables
substituted and translated into a system of polynomials the axiom
consists of the following polynomial equations.
\begin{align}
  \label[axiom]{eq:ax1}
  1 - \sum_{j \in [n]} \isvar(v,a,j) &= 0\\
  \label[axiom]{eq:ax2}
  \isvar(v,a,i) \cdot \isvar(v,a,j) &= 0,
  \text{ for~}i \neq j\\
  \nonumber
  \left(\sum_{j \le i} \isvar(v,a,j)\right)
  \left(1 - \sum_{j < i} \isvar(v,a,j)\right)\cdot\quad&\\
  \label[axiom]{eq:ax3}
  \overline{\isvar(v,a,i)} &= 0,
  \text{ for~}i \in [n]\\
  \label[axiom]{eq:ax4}
  \left(1 - \sum_{j \le i} \isvar(v,a,j)\right)
  \left(\sum_{j < i} \isvar(v,a,j)\right) &= 0,
  \text{ for~}i \in [n]\\
  \label[axiom]{eq:ax5}
  \left(1 - \sum_{j \le i} \isvar(v,a,j)\right)
  \isvar(v,a,i) &= 0,
  \text{ for~}i \in [n] \eqperiod
\end{align}
\Cref{eq:ax1} is equal to \cref{eq:ax-in-gt} and similarly
\cref{eq:ax2} is equal to \cref{eq:prod-zero-in}. In the following we
show that \cref{eq:ax3,eq:ax4,eq:ax5} can be derived from
\cref{eq:prod-zero-in}, the Boolean axioms and the negation axioms in
constant degree.

Consider \cref{eq:ax3}. Expand and rewrite modulo the Boolean axioms
and the negation axiom to obtain
% \begin{align}
%   \sum_{j \le i} \isvar(v,a,j) -
%   \isvar(v,a,i)\left(\sum_{j \le i} \isvar(v,a,j)\right) -\\
%   \left(\sum_{j \le i} \isvar(v,a,j)\right)
%   \left(\sum_{j < i} \isvar(v,a,j)\right) +\\
%   \isvar(v,a,i)
%   \left(\sum_{j \le i} \isvar(v,a,j)\right)
%   \left(\sum_{j < i} \isvar(v,a,j)\right)
%    = 0 \eqperiod
% \end{align}
% Rewrite modulo the Boolean axioms to get to
% \begin{align}
%   \sum_{j < i} \isvar(v,a,j) -
%   \isvar(v,a,i)\left(\sum_{j < i} \isvar(v,a,j)\right) -\\
%   \sum_{j < i} \isvar(v,a,j) -
%   2\sum_{j < j' < i}\isvar(v,a,j)\isvar(v,a,j') -
%   \isvar(v,a,i)\sum_{j < i}\isvar(v,a,j) +\\
%   \isvar(v,a,i)
%   \left(\sum_{j < i} \isvar(v,a,j)\right) +
%   \isvar(v,a,i)
%   \left(\sum_{j < i} \isvar(v,a,j)\right)^2
%   = 0 \eqperiod
% \end{align}
% We can cancel some terms and are left with the polynomial equation
\begin{align}
  &\isvar(v,a,i)
  \left(\sum_{j < i} \isvar(v,a,j)\right)^2-\nonumber\\
  &\qquad\qquad 2\sum_{j < j' < i}\isvar(v,a,j)\cdot\isvar(v,a,j')~-\nonumber\\
  &\qquad\qquad\qquad\qquad\isvar(v,a,i)\sum_{j < i}\isvar(v,a,j)
  = 0 \eqperiod
\end{align}
Observe that every term $t$ left in this polynomial is of the form
$t = t' \cdot \isvar(v,a,j) \cdot \isvar(v,a,j')$, for some
$j \neq j'\in [i]$ and a term $t'$ of degree at most $1$. But this
means that every term is equal to $0$ modulo \cref{eq:prod-zero-in}
and we thus see that \cref{eq:ax3} can be derived in constant degree
from $\circuit_s(f)$.

Let us consider \cref{eq:ax4}. Rewrite modulo the Boolean axiom to
obtain
\begin{align}
  &\isvar(v,a,i)\sum_{j < i}\isvar(v,a,j) +\nonumber\\
  &\qquad 2 \sum_{j < j' < i}\isvar(v,a,j)\cdot \isvar(v,a,j')
  = 0 \eqperiod
\end{align}
All terms are of the form of \cref{eq:prod-zero-in} and we can thus
derive \cref{eq:ax4} from $\circuit_s(f)$ in constant degree.

Last, we need to consider \cref{eq:ax5}. Note that modulo the Boolean
axiom we obtain the polynomial equation
\begin{align}
  - \isvar(v,a,i)\sum_{j < i}\isvar(v,a,j) = 0 \eqperiod
\end{align}
Also in this polynomial every term is of the form of
\cref{eq:prod-zero-in} and thus also \cref{eq:ax5} can be derived in
constant degree.

What remains is to show that \cref{ax:gate-cnf} can be derived from
$\circuit_s(f)$ in constant degree. This can be checked analogous to
\cref{ax:input-cnf} and we thus omit it here.

We conclude that all axioms of $\restrict{\circuit^{\mathrm{CNF}}_s(f)}{\rho}$ can
be derived from $\circuit_s(f)$ in constant degree and thus a degree
$d$ SoS refutation of $\circuit_s^{\mathrm{CNF}}(f)$ gives rise to a degree
$O(d)$ SoS refutation of $\circuit_s(f)$. Equivalently, a degree $d$
lower bound for $\circuit_s(f)$ implies a degree $\Omega(d)$ lower
bound for $\circuit_s^{\mathrm{CNF}}(f)$ as claimed.

%% file: paper.bbl
\newcommand{\etalchar}[1]{$^{#1}$}
\begin{thebibliography}{KMOW17}

\bibitem[ABRW04]{ABRW04Pseudorandom}
Michael Alekhnovich, Eli {Ben-Sasson}, Alexander~A. Razborov, and Avi
  Wigderson.
\newblock Pseudorandom generators in propositional proof complexity.
\newblock {\em SIAM Journal on Computing}, 34(1):67\nobreakdash--88, 2004.
\newblock Preliminary version in \emph{FOCS~'00}.

\bibitem[AH19]{AH18SosTradeoff}
Albert Atserias and Tuomas Hakoniemi.
\newblock {Size-Degree Trade-Offs for Sums-of-Squares and Positivstellensatz
  Proofs}.
\newblock In Amir Shpilka, editor, {\em 34th Computational Complexity
  Conference (CCC 2019)}, volume 137 of {\em Leibniz International Proceedings
  in Informatics (LIPIcs)}, pages 24:1--24:20, Dagstuhl, Germany, 2019. Schloss
  Dagstuhl--Leibniz-Zentrum fuer Informatik.

\bibitem[AOW15]{AOW15}
S.~R. Allen, R.~ODonnell, and D.~Witmer.
\newblock How to refute a random csp.
\newblock In {\em 2015 IEEE 56th Annual Symposium on Foundations of Computer
  Science (FOCS)}, pages 689--708, Los Alamitos, CA, USA, oct 2015. IEEE
  Computer Society.

\bibitem[Ber82]{B82}
S.~J. Berkowitz.
\newblock On some relationships between monotone and non-monotone circuit
  complexity.
\newblock Technical report, Technical Report, University of Toronto, 1982.

\bibitem[BHK{\etalchar{+}}16]{BHKKMP16clique}
B.~{Barak}, S.~B. {Hopkins}, J.~{Kelner}, P.~{Kothari}, A.~{Moitra}, and
  A.~{Potechin}.
\newblock A nearly tight sum-of-squares lower bound for the planted clique
  problem.
\newblock In {\em 2016 IEEE 57th Annual Symposium on Foundations of Computer
  Science (FOCS)}, pages 428--437, 2016.

\bibitem[BT06]{BT06}
Andrej Bogdanov and Luca Trevisan.
\newblock Average-case complexity.
\newblock {\em Foundations and Trends in Theoretical Computer Science},
  2(1):1--106, 2006.

\bibitem[BW06]{BW06}
Amos Beimel and Enav Weinreb.
\newblock Monotone circuits for monotone weighted threshold functions.
\newblock {\em Inf. Process. Lett.}, 97(1):12–18, jan 2006.

\bibitem[GHP02]{GHA02proofs}
Dima Grigoriev, Edward~A. Hirsch, and Dmitrii~V. Pasechnik.
\newblock Complexity of semi-algebraic proofs.
\newblock In {\em STACS 2002}, pages 419--430, Berlin, Heidelberg, 2002.
  Springer Berlin Heidelberg.

\bibitem[Gol20]{Goldreich2020}
Oded Goldreich.
\newblock {\em On (Valiant's) Polynomial-Size Monotone Formula for Majority},
  pages 17--23.
\newblock Springer International Publishing, Cham, 2020.

\bibitem[Gri01]{gri01xor}
Dima Grigoriev.
\newblock Linear lower bound on degrees of positivstellensatz calculus proofs
  for the parity.
\newblock {\em Theoretical Computer Science}, 259(1):613 -- 622, 2001.

\bibitem[GUV09]{GUV09Unbalanced}
Venkatesan Guruswami, Christopher Umans, and Salil Vadhan.
\newblock Unbalanced expanders and randomness extractors from
  {P}arvaresh--{V}ardy codes.
\newblock {\em Journal of the ACM}, 56(4):20:1\nobreakdash--20:34, July 2009.
\newblock Preliminary version in \emph{CCC~'07}.

\bibitem[GW95]{GW95}
Michel~X. Goemans and David~P. Williamson.
\newblock Improved approximation algorithms for maximum cut and satisfiability
  problems using semidefinite programming.
\newblock {\em J. {ACM}}, 42(6):1115--1145, 1995.

\bibitem[Hir18]{Hi18}
Shuichi Hirahara.
\newblock Non-black-box worst-case to average-case reductions within np.
\newblock In {\em 2018 IEEE 59th Annual Symposium on Foundations of Computer
  Science (FOCS)}, pages 247--258, 2018.

\bibitem[KC00]{KC00}
Valentine Kabanets and Jin-Yi Cai.
\newblock Circuit minimization problem.
\newblock In {\em Proceedings of the Thirty-Second Annual ACM Symposium on
  Theory of Computing}, STOC '00, page 73–79, New York, NY, USA, 2000.
  Association for Computing Machinery.

\bibitem[KMOW17]{kmow17anycsp}
Pravesh~K. Kothari, Ryuhei Mori, Ryan O’Donnell, and David Witmer.
\newblock Sum of squares lower bounds for refuting any csp.
\newblock In {\em Proceedings of the 49th Annual ACM SIGACT Symposium on Theory
  of Computing}, STOC 2017, page 132–145, New York, NY, USA, 2017.
  Association for Computing Machinery.

\bibitem[KMS98]{KMS98}
David Karger, Rajeev Motwani, and Madhu Sudan.
\newblock Approximate graph coloring by semidefinite programming.
\newblock {\em J. ACM}, 45(2):246–265, mar 1998.

\bibitem[Kra01]{Krajicek01}
Jan Krajíček.
\newblock On the weak pigeonhole principle.
\newblock {\em Fundamenta Mathematicae}, 170(1-2):123--140, 2001.

\bibitem[MPW15]{MPW15SumOfSquaresPlantedClique}
Raghu Meka, Aaron Potechin, and Avi Wigderson.
\newblock Sum-of-squares lower bounds for planted clique.
\newblock In {\em Proceedings of the 47th Annual ACM Symposium on Theory of
  Computing ({STOC}~'15)}, pages 87\nobreakdash--96, June 2015.

\bibitem[MW15]{MW15}
Cody~D. Murrayand and R.~Ryan Williams.
\newblock {On the (Non) NP-Hardness of Computing Circuit Complexity}.
\newblock In David Zuckerman, editor, {\em 30th Conference on Computational
  Complexity (CCC 2015)}, volume~33 of {\em Leibniz International Proceedings
  in Informatics (LIPIcs)}, pages 365--380, Dagstuhl, Germany, 2015. Schloss
  Dagstuhl--Leibniz-Zentrum fuer Informatik.

\bibitem[Pot20]{potechin20}
Aaron Potechin.
\newblock {Sum of Squares Bounds for the Ordering Principle}.
\newblock In Shubhangi Saraf, editor, {\em 35th Computational Complexity
  Conference (CCC 2020)}, volume 169 of {\em Leibniz International Proceedings
  in Informatics (LIPIcs)}, pages 38:1--38:37, Dagstuhl, Germany, 2020. Schloss
  Dagstuhl--Leibniz-Zentrum f{\"u}r Informatik.

\bibitem[PR04]{PitRaz04}
Toniann Pitassi and Ran Raz.
\newblock Regular resolution lower bounds for the weak pigeonhole principle.
\newblock {\em Combinatorica}, 24(3):503--524, 2004.
\newblock Preliminary version in \emph{STOC~'01}.

\bibitem[Raz98]{Razborov98}
Alexander~A. Razborov.
\newblock Lower bounds for the polynomial calculus.
\newblock {\em Computational Complexity}, 7(4):291\nobreakdash--324, December
  1998.

\bibitem[Raz04a]{Raz04}
Ran Raz.
\newblock Resolution lower bounds for the weak pigeonhole principle.
\newblock {\em J. {ACM}}, 51(2):115--138, 2004.

\bibitem[Raz04b]{Razborov04ResolutionLowerBoundsPM}
Alexander~A. Razborov.
\newblock Resolution lower bounds for perfect matching principles.
\newblock {\em Journal of Computer and System Sciences},
  69(1):3\nobreakdash--27, August 2004.
\newblock Preliminary version in \emph{CCC~'02}.

\bibitem[Raz15]{Razborov15PRG}
Alexander~A. Razborov.
\newblock Pseudorandom generators hard for \mbox{$k$-{DNF}} resolution and
  polynomial calculus resolution.
\newblock {\em Annals of Mathematics}, 181(2):415\nobreakdash--472, March 2015.

\bibitem[Raz21]{razborov21youtube}
Alexander Razborov.
\newblock {P}, {NP} and {Proof Complexity}.
\newblock \url{https://youtu.be/ZVL_HsPC4xE?t=2646}, 2021.
\newblock Accessed April 2022.

\bibitem[Raz22]{razborov22website}
Alexander Razborov.
\newblock Open problems.
\newblock \url{https://people.cs.uchicago.edu/~razborov/teaching/index.html},
  2022.
\newblock Accessed April 2022.

\bibitem[RRS17]{RRS17}
Prasad Raghavendra, Satish Rao, and Tselil Schramm.
\newblock Strongly refuting random csps below the spectral threshold.
\newblock In {\em Proceedings of the 49th Annual ACM SIGACT Symposium on Theory
  of Computing}, STOC 2017, page 121–131, New York, NY, USA, 2017.
  Association for Computing Machinery.

\bibitem[RWY02]{RazbWigYao02}
Alexander~A. Razborov, Avi Wigderson, and Andrew~Chi{-}Chih Yao.
\newblock Read-once branching programs, rectangular proofs of the pigeonhole
  principle and the transversal calculus.
\newblock {\em Combinatorica}, 22(4):555--574, 2002.
\newblock Preliminary version in \emph{STOC~'97}.

\bibitem[Val84]{Valiant84}
Leslie~G. Valiant.
\newblock Short monotone formulae for the majority function.
\newblock {\em J. Algorithms}, 5:363--366, 1984.

\end{thebibliography}
